\documentclass[reqno,11pt]{amsart}
\usepackage{fullpage}
\usepackage{amsfonts}
\usepackage{amssymb}
\usepackage{graphicx}
\usepackage{arydshln}
\usepackage{latexsym}
\usepackage[mathscr]{eucal}
\usepackage{cite}
\usepackage{bm}
\usepackage{mathrsfs}
\usepackage{color}
\usepackage{cases}
\usepackage{amsmath,amsthm,accents}
\numberwithin{equation}{section}
\DeclareMathOperator{\sech}{sech}
\DeclareMathOperator{\csch}{csch}
\usepackage{float}
\usepackage{diagbox}
\usepackage{array}
\usepackage{caption}
\usepackage{subfigure}
\numberwithin{equation}{section}
\newtheorem{Prop}{Proposition}
\newtheorem{Thm}{Theorem}

\def \tyb#1{\hbox{\tiny{[{\it{#1}}]}}}
\def \ty#1{\hbox{\tiny{{\it{#1}}}}}
\def \wh#1{\widehat{#1}}
\def \wt#1{\widetilde{#1}}

\newcommand{\nn}{\nonumber}
\newcommand{\bA}{\boldsymbol{A}}
\newcommand{\ba}{\boldsymbol{a}}

\newcommand{\bC}{\boldsymbol{C}}

\newcommand{\bE}{\boldsymbol{E}}
\newcommand{\bF}{\boldsymbol{F}}
\newcommand{\bG}{\boldsymbol{G}}
\newcommand{\bH}{\boldsymbol{H}}
\newcommand{\bI}{\boldsymbol{I}}
\newcommand{\bK}{\boldsymbol{K}}

\newcommand{\bM}{\boldsymbol{M}}

\newcommand{\br}{\boldsymbol{r}}
\newcommand{\bS}{\boldsymbol{S}}

\newcommand{\bT}{\boldsymbol{T}}
\newcommand{\bu}{\boldsymbol{u}}
\newcommand{\bv}{\boldsymbol{v}}

\newcommand{\bw}{\boldsymbol{w}}
\newcommand{\Ga}{\boldsymbol{\Gamma}}

\newcommand{\st}{\hbox{\tiny\it{T}}}

\newcommand{\tc}{\,^t \hskip -2pt {\boldsymbol{c}}}

\def \cd#1{\accentset{\circ}{#1}}

\begin{document}

\title{Solutions of local and nonlocal discrete complex modified Korteweg-de Vries equations and continuum limits}

\author{Ya-Nan Hu, Shou-Feng Shen, Song-lin Zhao$^{*}$\\
\\\lowercase{\scshape{Department of Applied Mathematics, Zhejiang University of Technology,
Hangzhou 310023, P.R. China}}}
\email{*Corresponding Author: songlinzhao@zjut.edu.cn}

\begin{abstract}

Cauchy matrix approach for the discrete Ablowitz-Kaup-Newell-Segur equations is reconsidered,
where two `proper' discrete Ablowitz-Kaup-Newell-Segur equations and two `unproper' discrete Ablowitz-Kaup-Newell-Segur equations
 are derived. The `proper' equations admit local reduction, while the `unproper' equations admit nonlocal reduction.
By imposing the local and nonlocal complex reductions on the obtained discrete Ablowitz-Kaup-Newell-Segur equations, two local and nonlocal
discrete complex modified Korteweg-de Vries equations are constructed. For the obtained local and nonlocal
discrete complex modified Korteweg-de Vries equations, soliton solutions and Jordan-block solutions
are presented by solving the determining equation set.
The dynamical behaviors of 1-soliton solution are analyzed and illustrated. Continuum limits of the resulting local and nonlocal
 discrete complex modified Korteweg-de Vries equations are discussed.

\end{abstract}


\maketitle

\section{Introduction}

In recent decade, the study of nonlocal nonlinear integrable equations lies at the forefront of research
in mathematical physics. This is the case since they
have been recognized as basic models for describing parity-charge-time symmetry in
quantum chromodynamics \cite{MPW}, electric circuits \cite{LSEK}, optics \cite{RMGCSK, MMGC},
Bose-Einstein condensates \cite{DGPS}, Alice-Bob events \cite{Lou-JMP,Lou-CTP}, and so forth. The first such equation, is
an integrable nonlocal nonlinear Schr\"{o}dinger equation
\begin{align}
\label{nl-NLS}
i\zeta_t+\zeta_{xx}+2\delta \zeta^2\zeta^*(-x,t)=0, \quad \delta=\pm 1,
\end{align}
proposed by Ablowitz and Musslimani \cite{AbMu-2013} as a special reduction of the
Ablowitz-Kaup-Newell-Segur (AKNS) hierarchy \cite{AKNS-1974}, where and whereafter
$i$ is the imaginary unit, and the asterisk $*$
represents complex conjugation. The equation \eqref{nl-NLS}
has parity-time symmetry owing to the self-induced potential
$V(x,t)=\zeta(x,t)\zeta^*(-x,t)=V^*(-x,t)$,
and the corresponding solution stated at distant locations $x$ and $-x$ are directly
coupled, reminiscent of quantum entanglement between pairs of particles.
From a mathematical point of view, studies of the nonlocal
equations are also interesting because these equations often feature in
distinctive types of solution behaviors, such as finite-time
solution blowup \cite{AbMu-2013,Yang-RW-NLS}, the simultaneous existence of solitons
and kinks \cite{JZ-JMAA}, the simultaneous existence of bright/dark solitons \cite{AbMu-2013,ALM-JMP},
and distinctive multisoliton patterns \cite{Yang-PLA}.

Following the introduction of this nonlocal nonlinear Schr\"{o}dinger equation, many other nonlocal integrable systems
have been proposed, one of which is the nonlocal complex modified Korteweg-de Vries (cmKdV) equation \cite{AM-mKdV}
\begin{align}
\label{nl-cmKdV}
U_t+U_{xxx}+24\delta UU^*(-x, -t)U_x=0, \quad \delta=\pm 1,
\end{align}
where the nonlocality is of reverse-space-time type.
In the past few years, the nonlocal cmKdV type equations started
to attract a lot of attention since
the local cmKdV type equations themselves have a wide range of physical applications in the propagation transverse-magnetic waves in nematic
optical fibers \cite{RREFM} and few-cycle optical pulses \cite{HWLPE}. Many methods have been developed to
construct a variety of exact solutions for the nonlocal cmKdV type equations.
In \cite{YY-SAPM}, Yang and Yang introduced a simple variable transformation to construct multisoliton and quasiperiodic solutions for
the nonlocal cmKdV equation from those of local cmKdV equation. With the help of
the Darboux transformations, various solutions with physical significance of the nonlocal cmKdV equation
were obtained in \cite{MSZ}, such as dark solitons, W-type solitons, M-type solitons, and periodic solutions.
Based on the double Wronskian solutions for the AKNS hierarchy,
bilinearization reduction scheme was developed by imposing a constraint
on the two basic vectors in double Wronskians such that two potential functions in the AKNS hierarchy
obey some nonlocal relations, which allowed us to obtain solutions of the nonlocal cmKdV equation \cite{CDLZ}.
Inverse scattering transform was developed to deal with the multisoliton solutions for a nonlocal cmKdV equation with nonzero boundary conditions
at infinity and constant phase shift \cite{Luo-IST}. Then
Riemann-Hilbert technique was applied to study soliton solutions for a generalized nonlocal cmKdV equation,
which includes the equation \eqref{nl-cmKdV} as a special case \cite{Ma-JGP}.
Hirota method and an improved Hirota bilinear method were also proposed to discuss soliton solutions for
the nonlocal cmKdV equation \cite{GP,YDL}.

Almost simultaneously, there has been rapid development of research on nonlocal integrable semi-discrete
nonlinear systems described by differential-difference equations,
involving exact solutions and dynamical behaviors \cite{AM-2014,DLZ-AMC,FZS-IJMPB,CNY,SMMC},
gauge equivalence \cite{MZ-JMP,MSZ-AML}, as well as integrability \cite{Ger,ALM-Nonl}.
These studies have tremendously increased our understanding about those particularly rich systems.
In recent years, the study of integrable discrete nonlinear systems governed by difference equations (lattice equations)
has made clear that these objects are not only interesting for their own sake but are, in fact,
more fundamental than their continuous and semi-discrete counterparts (see \cite{HJN-2016} and the references therein).
Quite expectedly, the properties of the nonlocal integrable continuous/semi-discrete nonlinear equations find
themselves reflected in the properties of their discrete analogues. The study of nonlocal integrable discrete
nonlinear systems originates from the proposal of nonlocal Adler-Bobenko-Suris (ABS) lattice equations \cite{ZKZ-SIGMA}.
For the whole equations in the two-component ABS lattice list \cite{BHQK}, they admit
reverse-$(n,m)$ and reverse-$n$ nonlocal reductions, respectively. In a series of works \cite{XFZ-TMPH,XZS-SAPM,ZXS-MMAS},
we introduce the Cauchy matrix and bilinearization reductions schemes for some nonlocal discrete integrable systems.
As a result, Cauchy matrix solutions to a discrete sine-Gordon equation and
double Casoratian solutions to a discrete sine-Gordon equation and a discrete mKdV equation were constructed.

These two methods have some resemblances. For instance, both methods are based on the solutions of before-reduced system.
In addition, they involve, first taking appropriate reductions to get the nonlocal integrable systems,
and second solving the matrix equation algebraically to derive the exact solutions. Even so, the bilinearization
reduction scheme is superior to the Cauchy matrix reduction scheme. This is because, in the Cauchy matrix reduction approach, solutions
of the original before-reduction AKNS system should satisfy two Sylvester equations \cite{XFZ-TMPH}. In
the real reduction case, the solvability of these two Sylvester equations usually conflicts with the
solvability of the Sylvester equation in the matrix equation set (for more detailed explanations,
one can refer to the conclusions in \cite{XZ-TMPH}).

In spite of the imperfection, Cauchy matrix reduction scheme is still an efficient approach
to investigate the solutions of nonlocal complex integrable systems \cite{XZ-TMPH,FZ-ROMP}. The Cauchy
matrix reduction scheme has closed connection with the so-called Cauchy matrix approach.
The latter one was firstly proposed by Nijhoff et al. to deal with the multisoliton solutions of the
ABS lattice equations \cite{NAJ-JPA} and subsequently extended by
Zhang and Zhao to obtain more kinds of exact solutions beyond soliton solutions \cite{ZZ-SAPM}.
The (generalized) Cauchy matrix approach is a byproduct of direct linearization \cite{FA,NQC}
and arises from the famous Sylvester equation \cite{Syl}. In the present paper, we
are interested in the solutions of nonlocal cmKdV equations reduced from the (third-order) discrete AKNS equations associated with the Cauchy
matrix framework introduced in \cite{ZS-ZNA}. Noting that the discrete AKNS equation given in \cite{ZS-ZNA}
doesn't admit nonlocal reduction, regardless of real reduction or complex reduction. Thus in this paper we will
reconsider the Cauchy matrix procedure of the discrete AKNS equation. We will show that there are four
discrete AKNS equations derived from the Cauchy matrix framework,
some of which admit local reduction, whereas the left ones admit reverse-$(n,m)$ nonlocal reduction. After imposing the local and
nonlocal complex reductions, two local and nonlocal discrete cmKdV equations are obtained. Solutions involving multisoliton
solutions and Jordan-block solutions for the resulting discrete cmKdV equations are discussed by solving the determining equation set (DES).
Continuum limits will be also discussed.

The paper is organized as follows. In Section 2, Cauchy matrix approach to the discrete AKNS equation is reestablished.
In Section 3, local and nonlocal complex reductions are imposed on the discrete AKNS equations and two
local and nonlocal cmKdV equations are listed. Exact solutions, including multisoliton
solutions and Jordan-block solutions, are constructed by solving the DES.
In Section 4, continuum limits are
introduced to discuss the local and nonlocal semi-discrete cmKdV equation and the local and nonlocal continuous
cmKdV equation. Section 5 is for conclusions.

\section{Cauchy matrix approach for discrete AKNS equation}

In this section we exhibit the Cauchy matrix framework for the discrete AKNS equation. We first
introduce a DES, including a Sylvester equation and two shift relations
of column vector $\br$. We then introduce master functions $\bS^{(i,j)}$ and list some properties.
As a consequence of combination of shift relations with the help of
recurrence relations and symmetric property of $\bS^{(i,j)}$, four discrete AKNS equations are revealed as closed-form.

We use the following convention for elementary shifts in the two directions of the lattice:
for a dependent variable $f$ defined on the two-dimensional lattice with discrete coordinates $(n,m)\in \mathbb{Z}^2$, e.g., $f=f(n,m)$,
the operations $f\rightarrow \wt{f}$ and $f\rightarrow \wh{f}$ denote elementary shifts in the two directions of the
lattice, i.e., $\wt{f}=f(n+1,m)$ and $\wh{f}=f(n,m+1)$, while for the combined shifts we have: $\wh{\wt{f}}=f(n+1,m+1)$.

\subsection{The DES and master functions $\bS^{(i,j)}$}

\underline{}We first introduce a DES
\begin{subequations}
\label{DES}
\begin{align}
& \label{SE}
\bK\bM-\bM\bK=\br\tc, \\
& \label{Shift-p}
(p\bI_{2N}-\bA\bK)\wt{\br}=(p\bI_{2N}+\bA\bK){\br}, \\
& \label{Shift-q}
(q\bI_{2N}-\bA\bK)\wh{\br}=(q\bI_{2N}+\bA\bK){\br},
\end{align}
\end{subequations}
with $(n,m)$-dependent functions $\br$ and $\bM$, and
$(n,m)$-independent nontrivial constant matrices $\bK$ and $\tc$,
in which matrices $\bA$, $\bK$, $\bM$, $\br$ and $\tc$ are of form
\begin{subequations}
\label{AKMrtc-def}
\begin{align}
& \bA=\left(
      \begin{array}{cc}
        \bI_{N_1} & \textbf{0}\\
        \textbf{0} & -\bI_{N_2} \\
        \end{array}
         \right),\quad
 \bK=\left(
        \begin{array}{cc}
          \bK_1 & \textbf{0}\\
          \textbf{0} & \bK_2 \\
        \end{array}
      \right),\quad
 \bM=\left(
        \begin{array}{cc}
          \textbf{0} & \bM_1 \\
          \bM_2 & \textbf{0} \\
        \end{array}
      \right),\\
& \br=\left(
        \begin{array}{cc}
          \br_1 & \textbf{0}\\
          \textbf{0} & \br_2 \\
        \end{array} \right),\quad
 \tc=\left(
        \begin{array}{cc}
          \textbf{0} & \tc_2 \\
          \tc_1 & \textbf{0} \\
        \end{array}
      \right),
\end{align}
\end{subequations}
where $\bK_i\in \mathbb{C}_{N_i\times N_i}$, $\bM_1\in \mathbb{C}_{N_1\times N_2}$,
$\bM_2\in \mathbb{C}_{N_2\times N_1}$, $\br_i \in \mathbb{C}_{N_i\times 1}$, $\tc_j \in \mathbb{C}_{1\times N_j}$ with $N_1+N_2=2N$.
Here and hereafter, $\bI_{N_j}$ means the $N_j$-th order unit matrix. By the definitions \eqref{AKMrtc-def},
we observe from the DES \eqref{DES} that
\begin{subequations}
\label{DES-comp}
\begin{align}
& \label{SE-M1}
\bK_1 \bM_1-\bM_1\bK_2=\br_1\, \tc_2, \\
& \label{SE-M2}
\bK_2 \bM_2-\bM_2\bK_1=\br_2\, \tc_1, \\
& (p\bI_{N_j}+(-1)^j\bK_j)\wt{\br}_j=(p\bI_{N_j}-(-1)^j\bK_j){\br}_j, \quad j=1,2, \\
& (q\bI_{N_j}+(-1)^j\bK_j)\wh{\br}_j=(q\bI_{N_j}-(-1)^j\bK_j){\br}_j, \quad j=1,2.
\end{align}
\end{subequations}
System \eqref{SE-M1} and \eqref{SE-M2} are the Sylvester equations and
have unique solution for $\bM_1$ and $\bM_2$ provided $\mathcal{E}(\bK_1)\bigcap \mathcal{E}(\bK_2)=\varnothing$,
where $\mathcal{E}(\bK_1)$ and $\mathcal{E}(\bK_2)$ denote the eigenvalue sets of $\bK_1$ and $\bK_2$, respectively.
In the rest parts of this section, we assume matrices $\bK_1$ and $\bK_2$ satisfy such a condition.
Besides, we also assume $1\notin \mathcal{E}(\bM_1\bM_2)$ and $|\bK|\neq0$.

To construct discrete AKNS equations in the Cauchy matrix framework, we introduce $2\times2$ matrix functions
\begin{align}
\label{Sij-def}
\bS^{(i,j)}=\tc \bK^j(\bI_{2N}+\bM)^{-1} \bK^i \br
=\left(\begin{array}{cc}
s^{(i,j)}_1 & s^{(i,j)}_2 \\
s^{(i,j)}_3 & s^{(i,j)}_4 \\
\end{array}
\right),\quad i,j\in\mathbb{Z},
\end{align}
whose components are expressed as
\begin{subequations}
\label{sij-ex}
\begin{align}
& \label{sij1-ex}
s^{(i,j)}_1=-\tc_2 \bK_2^{j}\bM_{2}(\bI_{N_1}-\bM_{1}\bM_{2})^{-1} \bK_1^{i} \br_1, \\
& \label{sij2-ex}
s^{(i,j)}_2=\tc_2 \bK_2^{j}(\bI_{N_2}-\bM_{2}\bM_{1})^{-1} \bK_2^{i} \br_2, \\
& \label{sij3-ex}
s^{(i,j)}_3=\tc_2 \bK_2^{j}(\bI_{N_2}-\bM_{2}\bM_{1})^{-1} \bK_2^{i} \br_2, \\
& \label{sij4-ex}
s^{(i,j)}_4=-\tc_1 \bK_1^{j}\bM_{1}(\bI_{N_2}-\bM_{2}\bM_{1})^{-1} \bK_2^{i} \br_2.
\end{align}
\end{subequations}
Since the matrix functions $\bS^{(i,j)}$ will be used to generate lattice equations, we call them as
master functions, which possess recurrence relations, shift relations, similarity invariance \cite{ZS-ZNA}
and symmetric property \cite{LQYZ-SAPM}. We list these four properties as follows.

\begin{Prop}(Recurrence relations)
For the master functions $\bS^{(i,j)}$ defined by \eqref{Sij-def} with $\bM,~\bK,~\br$ and $\tc$ satisfying the
Sylvester equation \eqref{SE}, the following relations hold
\begin{subequations}
\label{Sij-Recu}
\begin{align}
& \label{Sij-Recu-a}
\bS^{(i,j+s)}=\bS^{(i+s,j)}-\sum_{l=0}^{s-1}\bS^{(s-1-l,j)}\bS^{(i,l)}, \quad (s=1,2, \ldots), \\
& \label{Sij-Recu-b}
\bS^{(i,j-s)}=\bS^{(i-s,j)}+\sum_{l=-1}^{-s}\bS^{(-s-1-l,j)}\bS^{(i,l)}, \quad (s=1,2,\ldots).
\end{align}
\end{subequations}
In particular, when $s=1$ we have
\begin{subequations}
\label{Sij-Recu-s=1}
\begin{align}
& \label{Sij-Recu-s=1-a}
\bS^{(i,j+1)}=\bS^{(i+1,j)}-\bS^{(0,j)}\bS^{(i,0)}, \\
& \label{Sij-Recu-s=1-b}
\bS^{(i,j-1)}=\bS^{(i-1,j)}+\bS^{(-1,j)}\bS^{(i,-1)}.
\end{align}
\end{subequations}
\end{Prop}

\noindent\textbf{Remark 1.} Despite recurrence relations of master functions don't play vital role in the construction
of some integrable scalar lattice equations, such as ABS lattice equations \cite{NAJ-JPA,ZZ-SAPM} and lattice Boussinesq type equations
\cite{ZZN-SAPM}, they are indispensable in the construction of continuous KdV type equations \cite{Zhao-JNMP} and some discrete/continuous AKNS type
equations \cite{ZS-ZNA,Zhao-ROMP}.

\begin{Prop}(Shift relations)
For the master functions $\bS^{(i,j)}$ defined by \eqref{Sij-def} with $\bM,~\bK,~\br$ and $\tc$ satisfying the
DES \eqref{DES}, the following relations hold
\begin{subequations}
\label{Sij-shift}
\begin{align}
& \label{Sij-shift-a}
p\wt{\bS}^{(i,j)}+\ba\wt{\bS}^{(i,j+1)}=p\bS^{(i,j)}+\bS^{(i+1,j)}\ba-\bS^{(0,j)}\ba\wt{\bS}^{(i,0)}, \\
& \label{Sij-shift-b}
p\bS^{(i,j)}-\ba\bS^{(i,j+1)}=p\wt{\bS}^{(i,j)}-\wt{\bS}^{(i+1,j)}\ba+\wt{\bS}^{(0,j)}\ba\bS^{(i,0)}, \\
& \label{Sij-shift-c}
q\wh{\bS}^{(i,j)}+\ba\wh{\bS}^{(i,j+1)}=q\bS^{(i,j)}+\bS^{(i+1,j)}\ba-\bS^{(0,j)}\ba\wh{\bS}^{(i,0)}, \\
& \label{Sij-shift-d}
q\bS^{(i,j)}-\ba\bS^{(i,j+1)}=q\wh{\bS}^{(i,j)}-\wh{\bS}^{(i+1,j)}\ba+\wh{\bS}^{(0,j)}\ba\bS^{(i,0)},
\end{align}
\end{subequations}
where $\ba=\text{Diag}(1,-1)$.
\end{Prop}

\noindent\textbf{Remark 2.} Shift relations \eqref{Sij-shift} encode all the
information on the dynamics of the master functions $\bS^{(i,j)}$, w.r.t. the discrete
variables $n$ and $m$. In principle, one can derive closed-form lattice equations for individual elements chosen
from the $\bS^{(i,j)}$ as functions of the variables $n,~m$.

\begin{Prop}(Similarity invariance)
For the master functions $\bS^{(i,j)}$ defined by \eqref{Sij-def} with $\bM,~\bK,~\br$ and $\tc$ satisfying the
DES \eqref{DES}, they have the similarity invariance.
\end{Prop}
\begin{proof}
Suppose that under transform matrix $\bT=\text{Diag}(\bT_1,\bT_2)$, matrix $\bar{\bK}$ is similar to $\bK$, i.e.,
\begin{subequations}
\label{simi-trans}
\begin{align}
\label{K-barK}
\bar{\bK}=\bT \bK \bT^{-1}.
\end{align}
We denote
\begin{align}
\label{bar-Mrtc}
\bar{\bM}=\bT \bM \bT^{-1},\quad \bar{\br}=\bT\br,\quad \bar{\tc}=\tc\bT^{-1}.
\end{align}
\end{subequations}
Then one can easily know that
\begin{align}
\label{Sij-K-barK}
\bS^{(i,j)}=\tc \bK^j(\bI_{2N}+\bM)^{-1} \bK^i \br=\bar{\tc}\bar{\bK}^j(\bI_{2N}+\bar{\bM})^{-1} \bar{\bK}^i \bar{\br},
\end{align}
and the DES \eqref{DES} yields
\begin{subequations}
\label{DES-bar}
\begin{align}
& \label{SE-bar}
\bar{\bK}\bar{\bM}-\bar{\bM}\bar{\bK}=\bar{\br}\bar{\tc}, \\
& \label{Shift-p-bar}
(p\bI_{2N}-\bA\bar{\bK})\wt{\bar{\br}}=(p\bI_{2N}+\bA\bar{\bK}){\bar{\br}}, \\
& \label{Shift-q-bar}
(q\bI_{2N}-\bA\bar{\bK})\wh{\bar{\br}}=(q\bI_{2N}+\bA\bar{\bK}){\bar{\br}}.
\end{align}
\end{subequations}
Thus the master functions $\bS^{(i,j)}$ are invariant and the DES \eqref{DES} is covariant under similar
transformation \eqref{simi-trans}.
\end{proof}

\noindent\textbf{Remark 3.} Since $\bS^{(i,j)}$ are similarity invariant under transformation \eqref{simi-trans},
in what follows we take matrix $\bK$ as its Jordan canonical form $\Ga=\text{Diag}(\Ga_1,\Ga_2)$. In other words,
we can replace the DES \eqref{DES-comp} and master functions $\bS^{(i,j)}$ \eqref{Sij-def} by
\begin{subequations}
\label{DES-comp-JC}
\begin{align}
& \label{SE-M1-JC}
\Ga_1 \bM_1-\bM_1\Ga_2=\br_1\, \tc_2, \\
& \label{SE-M2-JC}
\Ga_2\bM_2-\bM_2\Ga_1=\br_2\, \tc_1, \\
& (p\bI_{N_j}+(-1)^j\Ga_j)\wt{\br}_j=(p\bI_{N_j}-(-1)^j\Ga_j){\br}_j, \quad j=1,2, \\
& (q\bI_{N_j}+(-1)^j\Ga_j)\wh{\br}_j=(q\bI_{N_j}-(-1)^j\Ga_j){\br}_j, \quad j=1,2,
\end{align}
\end{subequations}
as well as
\begin{align}
\label{Sij-JC}
\bS^{(i,j)}=\tc \Ga^j(\bI_{2N}+\bM)^{-1} \Ga^i \br,
\end{align}
where condition $\mathcal{E}(\Ga_1)\bigcap \mathcal{E}(\Ga_2)=\varnothing$ is assumed to guarantee the solvability of
the Sylvester equations \eqref{SE-M1-JC} and \eqref{SE-M2-JC}. Naturally, the recurrence relations \eqref{Sij-Recu}
and shift relations \eqref{Sij-shift} still hold.

\begin{Prop}\label{Prop-Symm}
(Symmetric property) For the master functions $\bS^{(i,j)}$ given by \eqref{Sij-JC} with $\bM,~\Ga,~\br$ and $\tc$ satisfying the
Sylvester equations \eqref{SE-M1-JC} and \eqref{SE-M2-JC}, the elements in $\bS^{(i,j)}$ and $\bS^{(j,i)}$ are related as the following
\begin{align}
\label{sij-Symm}
s_1^{(i,j)}=-s_4^{(j,i)}, \quad s_2^{(i,j)}=s_2^{(j,i)}, \quad s_3^{(i,j)}=s_3^{(j,i)}, \quad i,j \in\mathbb{Z},
\end{align}
i.e., $\bS^{(i,j)^{\st}}=-\sigma_2\bS^{(j,i)}\sigma_2$, where $\sigma_2$ is the second pauli matrix, and
$^{\st}$ means the transpose.
\end{Prop}

\noindent\textbf{Remark 4.} For the verification of Proposition \ref{Prop-Symm}, one can refer to \cite{LQYZ-SAPM}.
This proposition can be used to prove $\text{Det}(\bI_2-\bS^{(-1,0)})=1$, which is a core relation
in the construction of some AKNS type equations \cite{Zhao-JDEA,ZZS-TMPH}, and the studies of the
self-dual Yang-Mills equations \cite{LQYZ-SAPM,LQZ-PD}.

\subsection{The discrete AKNS equations}

Based on the symmetric property \eqref{sij-Symm}, we introduce the following new
variables based on the master functions $\bS^{(i,j)}$ as follows:
\begin{align}
\label{uvw-def}
\bu=\bS^{(0,0)}=\left(\begin{array}{cc}
u_1 & u_2 \\
u_3 & -u_1 \\
\end{array}
\right),\quad
\bv=\bI_2-\bS^{(-1,0)},\quad \bw=\bI_2+\bS^{(0,-1)},
\end{align}
where $u_j,~(j=1,2,3)$ are scalar functions. Here $\bv$ and $\bw$ are auxiliary functions and
variable $\bu$ is the object that we are interested in.

To proceed, shift relations \eqref{Sij-shift} with $i=j=0$ lead to
\begin{subequations}
\label{buS10S01-ma-shift}
\begin{align}
& p\wt\bu+\ba\wt{\bS}^{(0,1)}=p\bu+\bS^{(1,0)}\ba-\bu\ba\wt\bu, \\
& p\bu-\ba\bS^{(0,1)}=p\wt\bu-\wt{\bS}^{(1,0)}\ba+\wt\bu\ba\bu, \\
& q\wh\bu+\ba\wh{\bS}^{(0,1)}=q\bu+\bS^{(1,0)}\ba-\bu\ba\wh\bu, \\
& q\bu-\ba\bS^{(0,1)}=q\wh\bu-\wh{\bS}^{(1,0)}\ba+\wh\bu\ba\bu,
\end{align}
\end{subequations}
which moreover yield
\begin{align}
\label{u-shift}
(p-q)(\wh{\wt{\bu}}-\bu)+(\wt{\bu}-\wh{\bu})((p+q)\bI_2+\ba\bu-\ba\wh{\wt{\bu}})=0,
\end{align}
after eliminating variables $\bS^{(0,1)}$ and $\bS^{(1,0)}$. The entries $u_j,~(j=1,2,3)$ satisfy
\begin{subequations}
\label{buS10S01-shift}
\begin{align}
& \label{buS10S01-shift-a}
(p-q+\wh{u}_1-\wt{u}_1)(\wh{\wt{u}}_2-u_2)-(p+q+u_1-\wh{\wt{u}}_1)(\wh{u}_2-\wt{u}_2)=0, \\
& \label{buS10S01-shift-b}
(p-q+\wh{u}_1-\wt{u}_1)(\wh{\wt{u}}_3-u_3)-(p+q+u_1-\wh{\wt{u}}_1)(\wh{u}_3-\wt{u}_3)=0.
\end{align}
\end{subequations}

To seek other relations connecting these three entries, we take indices $i=-1, j=0$, respectively, $i=0, j=-1$
in \eqref{Sij-shift-a} and \eqref{Sij-shift-b} and have
\begin{subequations}
\label{bv-S11}
\begin{align}
p(\wt{\bv}-\bv)&=\ba\wt{\bS}^{(-1,1)}-\bu\ba\wt{\bv} \nn \\
 & =\ba\bS^{(-1,1)}-\wt{\bu}\ba\bv, \\
p(\wt{\bw}-\bw)&=\bS^{(1,-1)}\ba-\bw\ba\wt{\bu} \nn \\
 & =\wt{\bS}^{(1,-1)}\ba-\wt{\bw}\ba\wt{\bu}.
\end{align}
\end{subequations}
Inserting the recurrence relations (\eqref{Sij-Recu-s=1-a} with $i=0,j=-1$ and $i=-1,j=0$)
\begin{align}
\label{uvw-S11}
\bS^{(1,-1)}=\bw\bu, \quad \bS^{(-1,1)}=\bu\bv,
\end{align}
into \eqref{bv-S11} and noticing $\bv\bw=\bI_2$, we arrive at
\begin{subequations}
\label{bu-shift-pq}
\begin{align}
\label{bu-shift-p}
(p\bI_2+\bu\ba-\ba\wt{\bu})(p\bI_2+\ba\bu-\wt{\bu}\ba)=p^2\bI_2.
\end{align}
Obviously, the relation \eqref{bu-shift-p} holds also for the $\wh{\phantom{a}}$-shift,
simplify by replacing $p$ by $q$ and interchanging the roles of $n$ and $m$. Thus we have
\begin{align}
\label{bu-shift-q}
(q\bI_2+\bu\ba-\ba\wh{\bu})(q\bI_2+\ba\bu-\wh{\bu}\ba)=q^2\bI_2.
\end{align}
\end{subequations}
Taking \eqref{uvw-def} into \eqref{bu-shift-p} and \eqref{bu-shift-q} gives rise to
\begin{align}
\label{u1-albe}
p+u_1-\wt{u}_1=\alpha, \quad q+u_1-\wh{u}_1=\beta,
\end{align}
where
\begin{align}
\label{albe-def}
\alpha=\big(p^2-(u_2+\wt{u}_2)(u_3+\wt{u}_3)\big)^{\frac{1}{2}}, \quad
\beta=\big(q^2-(u_2+\wh{u}_2)(u_3+\wh{u}_3)\big)^{\frac{1}{2}}.
\end{align}
From \eqref{u1-albe}, we derive the following relations:
\begin{subequations}
\label{u1-albe-nu12}
\begin{align}
p-q+\wh{u}_1-\wt{u}_1&=\alpha-\beta=\wh{\alpha}-\wt{\beta}, \\
p+q+u_1-\wh{\wt{u}}_1&=\alpha+\wt{\beta}=\wh{\alpha}+\beta.
\end{align}
\end{subequations}
Taking \eqref{u1-albe-nu12} into \eqref{buS10S01-shift-a} yields
\begin{subequations}
\label{u2-albe-eps}
\begin{align}
& \label{u2-albe-eps-a}
(\alpha-\beta)(\wh{\wt{u}}_2-u_2)-(\alpha+\wt{\beta})(\wh{u}_2-\wt{u}_2)=0, \\
& \label{u2-albe-eps-b}
(\alpha-\beta)(\wh{\wt{u}}_2-u_2)-(\wh{\alpha}+\beta)(\wh{u}_2-\wt{u}_2)=0,\\
& \label{u2-albe-eps-c}
(\wh{\alpha}-\wt{\beta})(\wh{\wt{u}}_2-u_2)-(\alpha+\wt{\beta})(\wh{u}_2-\wt{u}_2)=0,\\
& \label{u2-albe-eps-d}
(\wh{\alpha}-\wt{\beta})(\wh{\wt{u}}_2-u_2)-(\wh{\alpha}+\beta)(\wh{u}_2-\wt{u}_2)=0.
\end{align}
\end{subequations}
In a similar fashion, substituting \eqref{u1-albe-nu12} into \eqref{buS10S01-shift-b}, we also get
\begin{subequations}
\label{u3-albe-eps}
\begin{align}
& \label{u3-albe-eps-a}
(\alpha-\beta)(\wh{\wt{u}}_3-u_3)-(\alpha+\wt{\beta})(\wh{u}_3-\wt{u}_3)=0, \\
& \label{u3-albe-eps-b}
(\alpha-\beta)(\wh{\wt{u}}_3-u_3)-(\wh{\alpha}+\beta)(\wh{u}_3-\wt{u}_3)=0,\\
& \label{u3-albe-eps-c}
(\wh{\alpha}-\wt{\beta})(\wh{\wt{u}}_3-u_3)-(\alpha+\wt{\beta})(\wh{u}_3-\wt{u}_3)=0,\\
& \label{u3-albe-eps-d}
(\wh{\alpha}-\wt{\beta})(\wh{\wt{u}}_3-u_3)-(\wh{\alpha}+\beta)(\wh{u}_3-\wt{u}_3)=0.
\end{align}
\end{subequations}

System (\eqref{u2-albe-eps-b}, \eqref{u3-albe-eps-b})
is nothing but the discrete AKNS equation given in \cite{ZS-ZNA}.
At first glance, we find that in \eqref{u2-albe-eps} and \eqref{u3-albe-eps} there is a considerable amount of redundancy.
For example, relation \eqref{u2-albe-eps-b} can be derived from \eqref{u2-albe-eps-a} by
replacing $p$ by $q$ whilst replacing $\wt{\phantom{a}}$-shift by $\wh{\phantom{a}}$-shift.
After removing the similar relations, we know that the remains are
\begin{subequations}
\label{u2-albe-eps-re}
\begin{align}
& \label{u2-albe-eps-re-a}
(\alpha-\beta)(\wh{\wt{u}}_2-u_2)-(\alpha+\wt{\beta})(\wh{u}_2-\wt{u}_2)=0, \\
& \label{u2-albe-eps-re-b}
(\wh{\alpha}-\wt{\beta})(\wh{\wt{u}}_2-u_2)-(\wh{\alpha}+\beta)(\wh{u}_2-\wt{u}_2)=0,
\end{align}
\end{subequations}
and
\begin{subequations}
\label{u3-albe-eps-re}
\begin{align}
& \label{u3-albe-eps-re-a}
(\alpha-\beta)(\wh{\wt{u}}_3-u_3)-(\alpha+\wt{\beta})(\wh{u}_3-\wt{u}_3)=0, \\
& \label{u3-albe-eps-re-b}
(\wh{\alpha}-\wt{\beta})(\wh{\wt{u}}_3-u_3)-(\wh{\alpha}+\beta)(\wh{u}_3-\wt{u}_3)=0.
\end{align}
\end{subequations}
Here we prefer to leaving system (\eqref{u2-albe-eps-d}, \eqref{u3-albe-eps-d})
rather than (\eqref{u2-albe-eps-c}, \eqref{u3-albe-eps-c}), since system (\eqref{u2-albe-eps-re-a},
\eqref{u3-albe-eps-re-b}) or (\eqref{u2-albe-eps-re-b},
\eqref{u3-albe-eps-re-a}) admits nonlocal reduction.
By combining the equations of \eqref{u2-albe-eps-re} and \eqref{u3-albe-eps-re}, there are four systems can be
obtained, i.e., systems (\eqref{u2-albe-eps-re-a}, \eqref{u3-albe-eps-re-a}), (\eqref{u2-albe-eps-re-a}, \eqref{u3-albe-eps-re-b}),
(\eqref{u2-albe-eps-re-b}, \eqref{u3-albe-eps-re-a}) and (\eqref{u2-albe-eps-re-b}, \eqref{u3-albe-eps-re-b}).
Among these systems, systems (\eqref{u2-albe-eps-re-a}, \eqref{u3-albe-eps-re-a})
and (\eqref{u2-albe-eps-re-b}, \eqref{u3-albe-eps-re-b}) are proper discretization of AKNS system since they
can reduce to the usual discrete mKdV equations by imposing local reductions $u_3=\delta u_2$ or $u_3=\delta u^*_2$ with
$\delta=\pm 1$. We view these two systems as `proper' discrete AKNS equations and the reserved two systems as `unproper'
discrete AKNS equations.
Cauchy matrix solution to the obtained discrete AKNS equations is given by
\begin{align}
\label{Solu-u2u3}
u_2=\tc_{2}(\bI_{N_2}-\bM_2\bM_1)^{-1}\br_2,
\quad u_3=\tc_{1}(\bI_{N_1}-\bM_1 \bM_2)^{-1}\br_1,
\end{align}
where $\{\bM_i,~\tc_i,~\br_i\}$ satisfy the canonical DES \eqref{DES-comp-JC}.

To summarize the structure obtained, we note that starting from the DES \eqref{DES},
depending dynamically on the lattice variables through the Sylvester equation \eqref{SE} and
the shift relations of column vector $\br$ given in \eqref{Shift-p} and \eqref{Shift-q},
we have defined $2\times 2$ matrix functions $\bS^{(i,j)}$, which possess recurrence relations,
shift relations, similarity invariance and symmetric property. From the shift relations of variable
$\bu$ \eqref{u-shift} and \eqref{bu-shift-pq}, discrete AKNS equations are derived
finally, some of which are `proper' discretization in the sense of local reduction.

\section{Local and nonlocal complex reductions of the discrete AKNS equations}

Despite the system (\eqref{u2-albe-eps-re-a}, \eqref{u3-albe-eps-re-b}) or (\eqref{u2-albe-eps-re-b}, \eqref{u3-albe-eps-re-a})
admits nonlocal real reduction, solutions for the reduced nonlocal real equation
can not be derived in the framework of Cauchy matrix reduction.
We in this section just aim to discuss local and nonlocal complex reductions
of the discrete AKNS equations obtained in the above section. For the sake of simplicity,
we denote $u_2=u$, $u_3=v$ and introduce the variables
\begin{align}
\label{munu-def}
\mu_{\sigma}=(p^2-\delta(u+\wt{u})(u^*_{\sigma}+\wt{u}^*_{\sigma}))^{\frac{1}{2}}, \quad
\nu_{\sigma}=(q^2-\delta(u+\wh{u})(u^*_{\sigma}+\wh{u}^*_{\sigma}))^{\frac{1}{2}}, \quad \sigma=\pm 1,
\end{align}
where we have used the notation $g_{\sigma}:=g(\sigma x_1,\sigma x_2)$ for the function $g:=g(x_1, x_2)$.
It is worth pointing out that for the shift operations, we have $\wt{g}_{\sigma}=g(\sigma (x_1+1),\sigma x_2)$,
$\wh{g}_{\sigma}=g(\sigma x_1,\sigma (x_2+1))$ and $\wh{\wt{g}}_{\sigma}=g(\sigma (x_1+1),\sigma (x_2+1))$.
We kindly remind the reader of the distinction between $u_{\sigma}$ as $\sigma=1$ and component $u_1$ in matrix.

\subsection{Reductions of the system (\eqref{u2-albe-eps-re}, \eqref{u3-albe-eps-re})}

The system (\eqref{u2-albe-eps-re}, \eqref{u3-albe-eps-re}) admits local and nonlocal
complex reductions $v=\delta u_{\sigma}^*$, and the resulting reduced equations are
\begin{align}
\label{lcmKdV-1}
(\mu_{\sigma}-\nu_{\sigma})(\wh{\wt{u}}-u)-(\mu_{\sigma}+\wt{\nu}_{\sigma})(\wh{u}-\wt{u})=0,
\end{align}
and
\begin{align}
\label{lcmKdV-2}
(\wh{\mu}_{\sigma}-\wt{\nu}_{\sigma})(\wh{\wt{u}}-u)-(\wh{\mu}_{\sigma}+\nu_{\sigma})(\wh{u}-\wt{u})=0.
\end{align}
In order to be more intuitive, we list the reduced cases in Table 1.
\begin{table}[H]
  \centering
  \caption{The reductions of the system (\eqref{u2-albe-eps-re}, \eqref{u3-albe-eps-re})}
  \label{tabel-1}
\begin{tabular}{|>{\centering\arraybackslash}m{2cm}|>{\centering\arraybackslash}m{2cm}|>
{\centering\arraybackslash}m{2cm}|>{\centering\arraybackslash}m{2cm}|>{\centering\arraybackslash}m{2cm}|}
   \cline{1-3}
   \diagbox[width=2.4cm]{\eqref{u3-albe-eps-re}}{\eqref{u2-albe-eps-re}}& \eqref{u2-albe-eps-re-a}& \eqref{u2-albe-eps-re-b}  \\
   \cline{1-3}
   \eqref{u3-albe-eps-re-a} & $v=\delta u^*$ & $v=\delta u_{-1}^*$ \\
   \cline{1-3}
   \eqref{u3-albe-eps-re-b} & $v=\delta u_{-1}^*$ & $v=\delta u^*$ \\
   \cline{1-3}
   \end{tabular}
\end{table}
When $\sigma=1$, both equations \eqref{lcmKdV-1} and \eqref{lcmKdV-2} are referred to as local discrete cmKdV equations.
While when $\sigma=-1$, these two equations are considered to be the nonlocal discrete cmKdV equations.
All the discrete cmKdV equations, both local and nonlocal, are preserved under
transformations $u\rightarrow -u$ and $u\rightarrow \pm iu$.

\subsection{Cauchy matrix solutions for equations \eqref{lcmKdV-1} and \eqref{lcmKdV-2}}

The key point in the construction of Cauchy matrix solutions to the local and nonlocal discrete cmKdV equations
is to establish the constraints of the elements $(\br_1, \tc_1, \bM_1)$ and $(\br_2, \tc_2, \bM_2)$
in the Cauchy matrix solution \eqref{Solu-u2u3} such that variables $u$ and $v$ satisfy $v=\delta u_{\sigma}^*$.
To this end, we assume $N_1=N_2=N$, namely matrices $\Ga_1$ and $\Ga_2$ have the same order.
For convenience, we denote $\bI_{N}=\bI$.

With regard to Cauchy matrix solutions of the local and nonlocal discrete cmKdV equations, we have the following result.
\begin{Thm}
\label{dcmKdV-solu-Thm}
The function
\begin{align}
\label{dcmKdV-solu}
u=\tc_{2}(\bI-\bM_2\bM_1)^{-1}\br_2
\end{align}
solves the local and nonlocal discrete cmKdV equation \eqref{lcmKdV-1} or \eqref{lcmKdV-2},
provided that the entities satisfy canonical DES \eqref{DES-comp-JC} and simultaneously obey the constraints
\begin{align}
\label{dcmKdV-M1M12}
\br_1=\varepsilon \bT \br^*_{2,\sigma}, \quad
\tc_{1}=\varepsilon \tc^{*}_{2}\bT^{-1}, \quad
\bM_1=-\delta\sigma\bT \bM^*_{2,\sigma}\bT^*,
\end{align}
in which $\bT\in \mathbb{C}_{N\times N}$ is a constant matrix satisfying
\begin{align}\label{dcmKdV-at-eq}
\Ga_1\bT+\sigma\bT\Ga^*_2=0,\quad \bC_1=\varepsilon\bT\bC_2^*,\quad \varepsilon^2=\varepsilon^{*^2}=\delta.
\end{align}
\end{Thm}
For the verification of this theorem, one can refer to \cite{XZ-TMPH}. Inserting \eqref{dcmKdV-M1M12} into \eqref{dcmKdV-solu}
and denoting $\cd{\bM}_2=\bM_2\bT$, we find that solution \eqref{dcmKdV-solu} can be rewritten as
\begin{align}
\label{dcmKdV-solu-u}
u=\tc_{2}(\bI+\delta\sigma\cd{\bM}_2\cd{\bM}^*_{2,\sigma})^{-1}\br_2,
\end{align}
where the entries $\{\tc_{2}, \cd\bM_2, \br_2\}$ satisfy the Jordan canonical DES
\begin{subequations}
\label{dcmKdV-so-nT}
\begin{align}
& \label{KM-rtc}
\Ga_2\cd{\bM}_2+\sigma\cd{\bM}_2\Ga^*_2=\varepsilon\br_2 \tc^{*}_{2}, \\
& \label{r2-p}
(p\bI+\Ga_2)\wt{\br}_2=(p\bI-\Ga_2){\br}_2, \\
& \label{r2-q}
(q\bI+\Ga_2)\wh{\br}_2=(q\bI-\Ga_2){\br}_2.
\end{align}
\end{subequations}
In order to give explicit expressions of $u$ in \eqref{dcmKdV-solu-u}, one just needs to solve the
matrix equations \eqref{dcmKdV-so-nT}. In this system, equations \eqref{r2-p} and \eqref{r2-q}
are used to give rise to column vector $\br_2$ and Sylvester equation \eqref{KM-rtc}
is applied to determine matrix $\cd{\bM}_2$. In terms of the linearity of \eqref{r2-p} and \eqref{r2-q}, one
recognizes that
\begin{align}
\br_2=(p\bI-\Ga_2)^{n}(p\bI+\Ga_2)^{-n}(q\bI-\Ga_2)^{m}(q\bI+\Ga_2)^{-m}\bC,
\end{align}
where $\bC$ is a constant column vector. With regard to the procedure of solving the Sylvester equation \eqref{KM-rtc}, one can
denote $\br_2=\bF\bE$, $\tc_2=\bE^{\st}\bH$ and $\cd{\bM}_2=\bF\bG\bH^*$, where $\bE$ is an appropriate column vector (cf. \cite{ZZ-SAPM}).

\subsubsection{Some notations}

Because $\Ga_2$ is of form canonical structure, it is possible to give a complete classification for the solutions.
We firstly introduce some notations, where usually the subscripts $_D$ and $_J$ correspond to
the cases of $\Ga_2$ being diagonal and being of Jordan-block, respectively. In what follows, we list some notations.
\begin{subequations}
\label{notations}
\begin{align}
\label{pwf}
\text{plane-wave~factors:} ~~
&\rho_j=\bigg(\frac{p-k_j}{p+k_j}\bigg)^n\bigg(\frac{q-k_j}{q+k_j}\bigg)^m\rho_j^0,~\mathrm{with~ constants~}\rho^0_j, \\
N\times N ~\mathrm{matrix:}~~
& \Ga_{\ty{D}}(\{k_j\}^{N}_{1})=\mathrm{Diag}(k_1, k_2, \ldots, k_{N}),
\end{align}
\begin{align}
N\times N ~\mathrm{matrix:}~~
& \Ga_{\ty{J}}(a)
=\left(\begin{array}{cccccc}
a & 0    & 0   & \cdots & 0   & 0 \\
1   & a  & 0   & \cdots & 0   & 0 \\
0   & 1  & a   & \cdots & 0   & 0 \\
\vdots &\vdots &\vdots &\vdots &\vdots &\vdots \\
0   & 0    & 0   & \cdots & 1   & a
\end{array}\right),\\
N\times N ~\mathrm{matrix:}~~
& \bF_{\ty{J}}(k_1)
=\left(
\begin{array}{ccccc}
\rho_1 & 0 & 0 & \cdots & 0\\
\frac{\partial_{k_1}\rho_1}{1!} & \rho_1 & 0 & \cdots & 0\\
\frac{\partial^{2}_{k_1}\rho_1}{2!} &\frac{\partial_{k_1}\rho_1}{1!} & \rho_1 & \cdots & 0\\
\vdots &\vdots &\vdots & \ddots & \vdots\\
\frac{\partial^{N-1}_{k_1}\rho_1}{(N-1)!} & \frac{\partial^{N-2}_{k_1}\rho_1 }{(N-2)!} & \frac{\partial^{N-3}_{k_1}\rho_1}{(N-3)!} & \cdots & \rho_1
\end{array}
\right),
\end{align}
\begin{align}
N\times N ~\mathrm{matrix:}~~
& \bH_{\ty{J}}(\{c_j\}^{N}_{1})
=\left(\begin{array}{ccccc}
c_1 & \cdots  & c_{N-2}  & c_{N-1} & c_N\\
c_2 & \cdots & c_{N-1}  & c_N & 0\\
c_3 &\cdots & c_N & 0 & 0\\
\vdots &\vdots & \vdots & \vdots & \vdots\\
c_N & \cdots & 0 & 0 & 0
\end{array}
\right),\\
N\text{-th~~column~~vector:}~~
& \bE_{\ty{D}}=(1, 1, \ldots, 1)^{\st}, \\
N\text{-th~~column~~vector:}~~
& \bE_{\ty{J}}=(1, 0, \ldots, 0)^{\st}, \\
N\times N ~\mathrm{matrix:}~~
& \bG_{\ty{D}}(\{k_i\}^{N}_{1})
=(a^{-1}_{ij})_{N \times N}, \\
N\times N ~\mathrm{matrix:}~~
& \bG_{\ty{J}}(k_1)=(g_{i,j})_{N\times N},
~~~g_{i,j}=\mathrm{C}^{i-1}_{i+j-2}(-1)^{i+j}\varepsilon\sigma^ja_{11}^{1-i-j},
\end{align}
\end{subequations}
where $a_{ij}=k_i+\sigma k^*_j$ and \[ \mathrm{C}^{i}_{j}=\frac{j!}{i!(j-i)!},\quad (j\geq i).\]

The $N$-th order matrix  in the following form
\begin{equation}
\mathcal{A}=\left(\begin{array}{cccccc}
a_0 & 0    & 0   & \cdots & 0   & 0 \\
a_1 & a_0  & 0   & \cdots & 0   & 0 \\
a_2 & a_1  & a_0 & \cdots & 0   & 0 \\
\vdots &\vdots &\cdots &\vdots &\vdots &\vdots \\
a_{N-1} & a_{N-2} & a_{N-3}  & \cdots &  a_1   & a_0
\end{array}\right)_{N\times N}
\label{A}
\end{equation}
with scalar elements $\{a_j\} \subset \mathbb{C}$ is a $N$-th order lower triangular Toeplitz matrix.
All such matrices compose a commutative set $\widetilde{G}^{\tyb{N}}$ with respect to matrix multiplication
and the subset
\[G^{\tyb{N}}=\big \{\mathcal{A} \big |~\big. \mathcal{A}\in \widetilde{G}^{\tyb{N}},~|\mathcal{A}|\neq 0 \big\}\]
is an Abelian group. Such kind of matrices play useful roles in the expression of exact solutions for soliton equations.
For more properties of such matrices one can refer to \cite{Z-KdV-2006}.

\subsubsection{Soliton solutions}

If $\Ga_2$ is a diagonal matrix as $\Ga_2=\Ga_{\ty{D}}(\{k_i\}^{N}_{1})$, then we have
\begin{align}
\label{Krtc-soli}
\br_2=\bF\cdot\bE_{\ty{D}}, \quad \tc_2=\bE^{\st}_{\ty{D}}\cdot\bH, \quad \cd{\bM}_2=\bF\cdot\bG\cdot\bH^*,
\end{align}
with $\bF=\Ga_{\ty{D}}(\{\rho_j\}^{N}_{1})$, $\bH=\Ga_{\ty{D}}(\{c_j\}^{N}_{1})$ and $\bG=\bG_{\ty{D}}(\{k_i\}^{N}_{1})$.
In this case, $u$ in \eqref{dcmKdV-solu-u} leads to multisoliton solutions.

In the case of $N=1$, we write down 1-soliton solution
\begin{align}
\label{u-soli-1}
u=\frac{c_1 a_{11}^2\rho_1}{a_{11}^2+\delta |c_1|^2\rho_1\rho^*_{1,\sigma}}.
\end{align}
The carrier wave in the local case ($\sigma=1$) reads
\begin{align}
\label{u-d-ca-I}
|u|^2=\Biggl\{
\begin{array}{ll}
\lambda^2\sech^2(\frac{1}{2}\ln (A_1^nB_1^m)+\ln C_1),& \text{with} \quad \delta=1, \\
\lambda^2\csch^2(\frac{1}{2}\ln (A_1^nB_1^m)+\ln C_1),& \text{with} \quad \delta=-1,\\
\end{array}
\end{align}
where we have taken $k_1=\lambda+i\chi$ and
\begin{align}
\label{ABC1}
A_1=\dfrac{(p-\lambda)^{2}+\chi^{2}}{(p+\lambda)^{2}+\chi^{2}}, \quad
B_1=\dfrac{(q-\lambda)^{2}+\chi^{2}}{(q+\lambda)^{2}+\chi^{2}}, \quad C_1=\frac{|c_1\rho_1^0|}{2|\lambda|}.
\end{align}

When $\delta=1$, \eqref{u-d-ca-I} describes a stable traveling wave with velocity $-\ln B_1/\ln A_1$
and amplitude $\lambda^2$, where the initial phase is $\ln C_1$. When $\delta=-1$, wave \eqref{u-d-ca-I}
has singularity along point trace $n=\ln (B_1^{-m}C_1^{-2})/\ln A_1$. We depict this soliton in Figure 1.

\begin{center}
\begin{picture}(120,100)
\put(-120,-23){\resizebox{!}{4cm}{\includegraphics{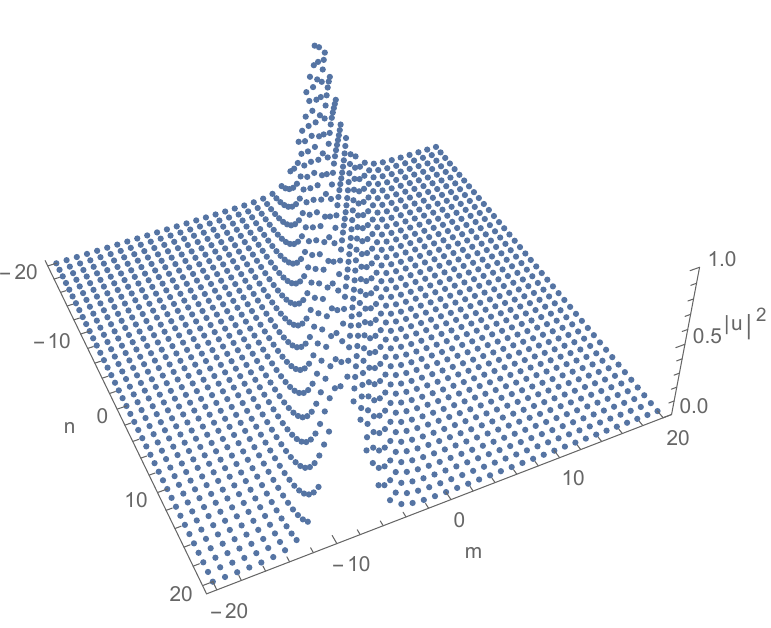}}}
\put(100,-23){\resizebox{!}{4cm}{\includegraphics{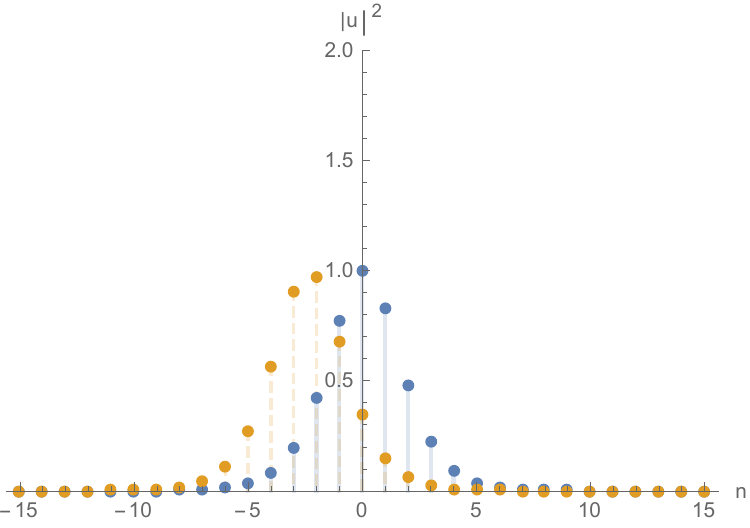}}}
\end{picture}
\end{center}
\vskip 20pt
\begin{center}
\begin{minipage}{15cm}{\footnotesize
\quad\qquad\qquad\qquad\qquad\qquad(a)\qquad\qquad\qquad\qquad\qquad\qquad\qquad\qquad\qquad\qquad\quad\quad\quad (b)}
\end{minipage}
\end{center}
\vskip 10pt
\begin{center}
\begin{picture}(120,80)
\put(-120,-23){\resizebox{!}{4cm}{\includegraphics{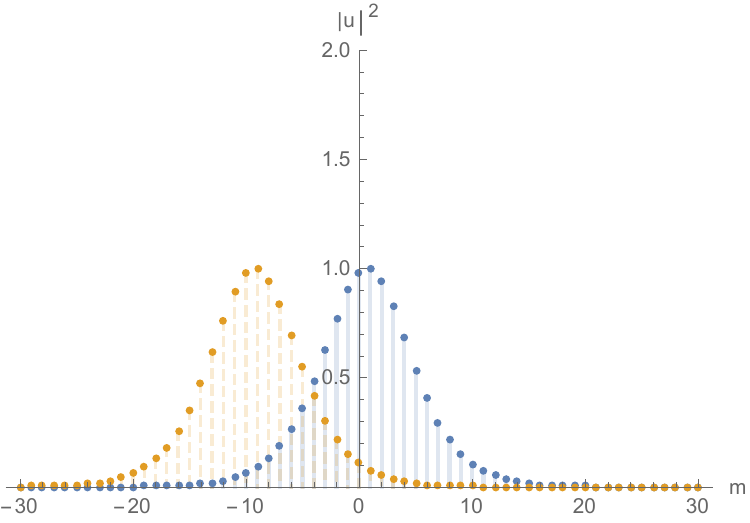}}}
\put(100,-23){\resizebox{!}{4cm}{\includegraphics{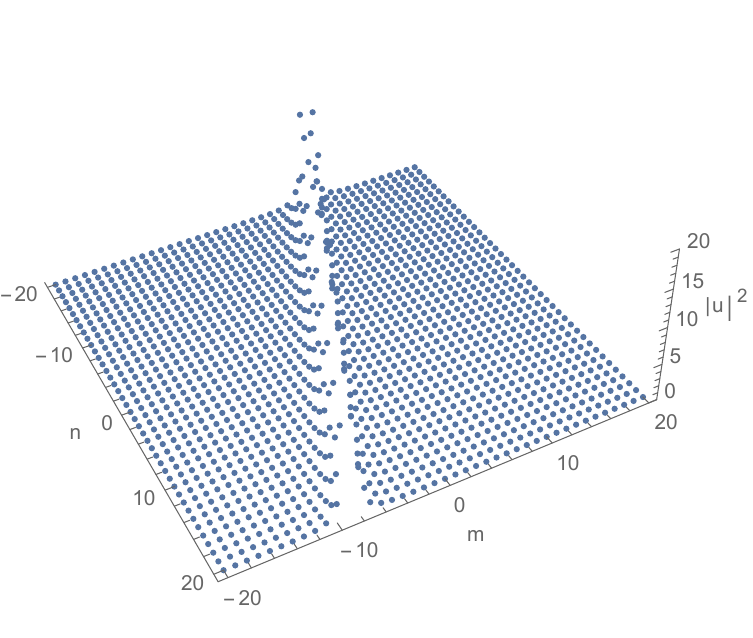}}}
\end{picture}
\end{center}
\vskip 20pt
\begin{center}
\begin{minipage}{15cm}{\footnotesize
\quad\qquad\qquad\qquad\qquad\qquad(c)\qquad\qquad\qquad\qquad\qquad\qquad\qquad\qquad\qquad\qquad\quad\quad\quad (d)\\
{\bf Fig. 1} Shape and motion with $|u|^2$ given by \eqref{u-d-ca-I} for $k_1=1+2i, p=2, q=0.5, \rho_1^0=1$ and $c_1=1+i$.
(a) 3D-plot for $\delta=1$.
(b) waves in blue and yellow stand for plot (a) at $m=-2$ and $m=4$, respectively.
(c) waves in blue and yellow stand for plot (a) at $n=-1$ and $n=3$, respectively.
(d) 3D-plot for $\delta=-1$.}
\end{minipage}
\end{center}

The wave package in the nonlocal case ($\sigma=-1$) has the form
\begin{align}
\label{u-d-ca-II}
|u|^2=\dfrac{4\chi^{2}A^n_1B^m_1}{C_2+C_2^{-1}-2\delta \cos(2n\arctan \theta_1+2m\arctan\theta_2)},
\end{align}
where $A_1$ and $B_1$ are defined by \eqref{ABC1} and
\begin{align}
& \label{theta12-def}
C_2=4\chi^{2}|c_1\rho_1^0|^{-2}, \quad \theta_1=\frac{2 p\chi}{\lambda^{2}+\chi^{2}-p^{2}}, \quad
\theta_2=\frac{2 q\chi}{\lambda^{2}+\chi^{2}-q^{2}}.
\end{align}
The solution \eqref{u-d-ca-II} has oscillatory phenomenon due to
the involvement of cosine function in denominator, which is singular for $C_2=1$
and nonsingular for $C_2\neq 1$. We illustrate soliton \eqref{u-d-ca-II} in Figure 2.
\begin{center}
\begin{picture}(120,100)
\put(-120,-23){\resizebox{!}{4cm}{\includegraphics{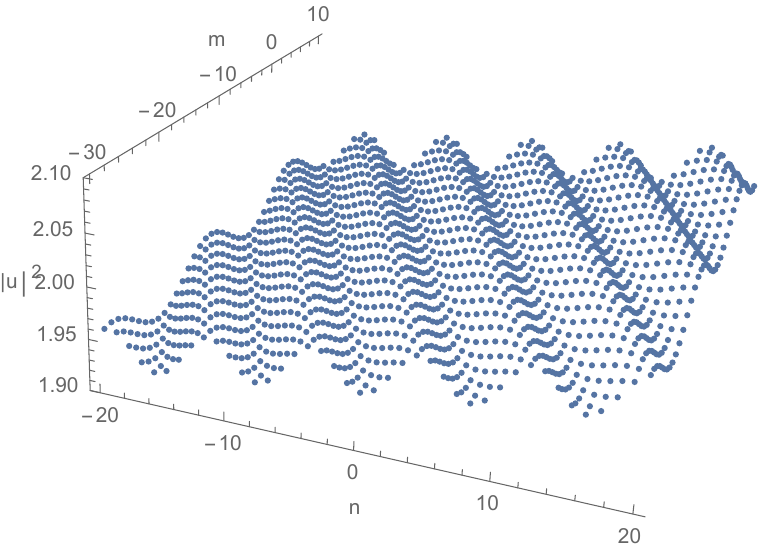}}}
\put(100,-23){\resizebox{!}{4cm}{\includegraphics{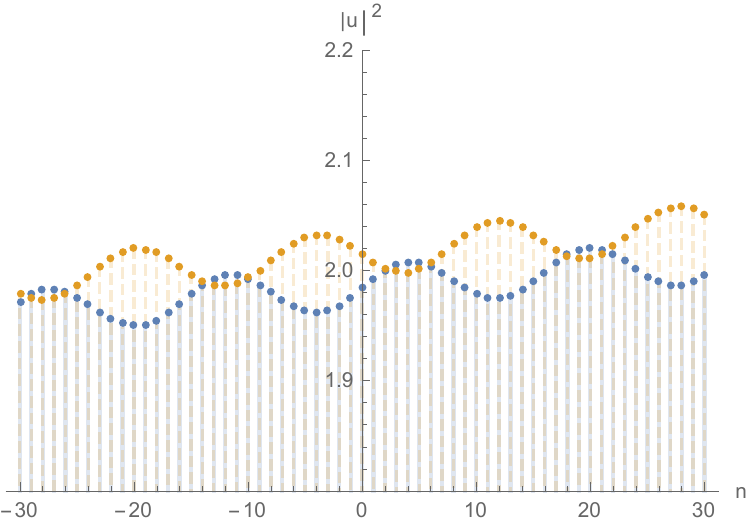}}}
\end{picture}
\end{center}
\vskip 20pt
\begin{center}
\begin{minipage}{15cm}{\footnotesize
\quad\qquad\qquad\qquad\qquad\qquad(a)\qquad\qquad\qquad\qquad\qquad\qquad\qquad\qquad\qquad\qquad\quad\quad\quad (b)}
\end{minipage}
\end{center}
\vskip 10pt
\begin{center}
\begin{picture}(120,80)
\put(-120,-23){\resizebox{!}{4cm}{\includegraphics{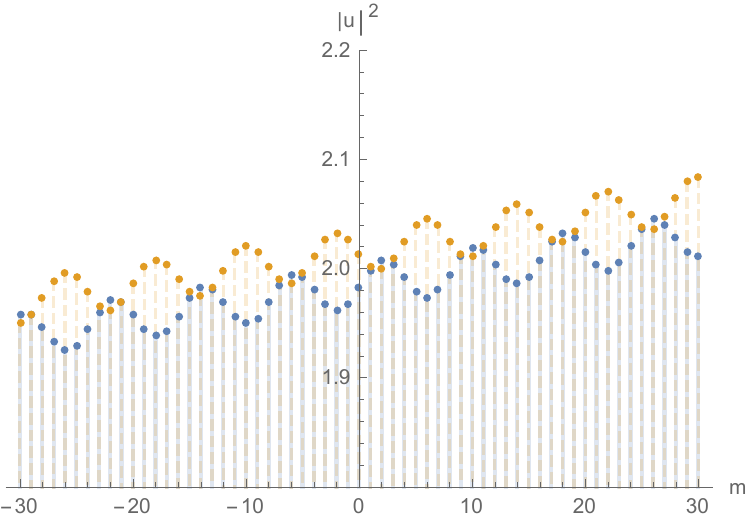}}}
\put(100,-23){\resizebox{!}{4cm}{\includegraphics{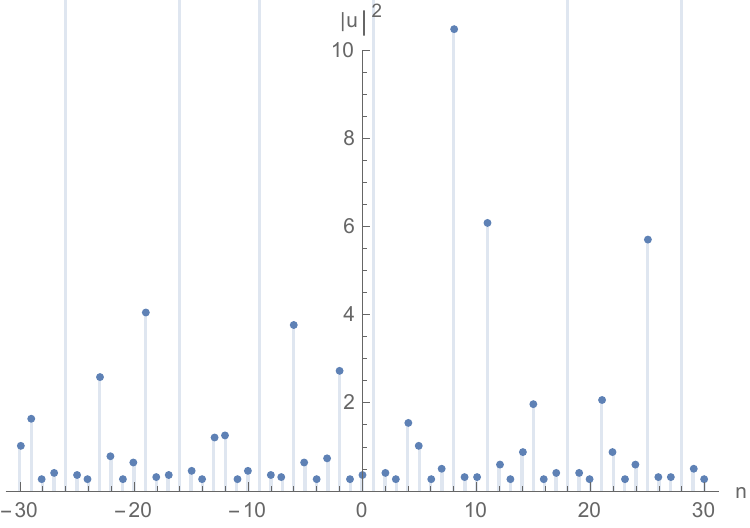}}}
\end{picture}
\end{center}
\vskip 20pt
\begin{center}
\begin{minipage}{15cm}{\footnotesize
\quad\qquad\qquad\qquad\qquad\qquad(c)\qquad\qquad\qquad\qquad\qquad\qquad\qquad\qquad\qquad\qquad\quad\quad\quad (d)\\
{\bf Fig. 2} Shape and motion with $|u|^2$ given by \eqref{u-d-ca-II} for $p=1,~q=2,~\rho_1^0=1$ and $\delta=1$.
(a) 3D-plot for $k_1=-0.01+10i$ and $c_1=1+i$.
(b) waves in blue and yellow stand for plot (a) at $m=-10$ and $m=10$, respectively.
(c) waves in blue and yellow stand for plot (a) at $n=-20$ and $n=20$, respectively.
(d) 2D-plot for $k_1=0.5i$ and $c_1=1$ at $m=-2$.}
\end{minipage}
\end{center}
The 2-soliton solutions reads $u=f_2/f_1$ with
\begin{subequations}
\label{u-soli-2}
\begin{align}
& f_{1}=1+\delta\sum^{2}_{i=1}\sum^{2}_{j=1}\dfrac{\Theta_i\Theta^*_{j,\sigma}}
{\varGamma_{ij}^{[ij]}}+\dfrac{\tau_{12}^{[12]}
(\varGamma_{12}^{[21]}-\varGamma_{11}^{[22]})^{2}}{(\varGamma_{12}^{[21]}\varGamma_{11}^{[22]})^{2}}, \\
& f_{2}=\sum_{i=1}^{2}\Theta_i+\delta\Theta_1\Theta_2
\sum_{j=1}^2\Theta^*_{j,\sigma}(\Lambda_{1j}^{[2j]}/\varGamma_{1j}^{[2j]})^2,
\end{align}
\end{subequations}
and the 3-soliton solutions are described as $u=g_2/g_1$ with
\begin{subequations}
\label{u-soli-3}
\begin{align}
& g_{1}=1+\delta\bigg(\sum^{3}_{i=1}\sum^{3}_{j=1}\dfrac{\Theta_i\Theta^*_{j,\sigma}}{\varGamma_{ij}^{[ij]}}
+\Upsilon\frac{\prod_{j=1}^{3}\Theta_j\Theta_{j,\sigma}^{*}}
{\prod^{3}_{i,j=1}\varGamma_{ij}^{[ij]}}\bigg)
+\sum_{1\leq i<j}^{3}\dfrac{\tau_{ij}^{[ij]}
(\varGamma_{ii}^{[jj]}-\varGamma_{ij}^{[ji]})^{2}}{(\varGamma_{ii}^{[jj]}\varGamma_{ij}^{[ji]})^{2}} \nn \\
&\qquad +\sum_{i=1,j\neq l\neq i}^3\frac{\tau_{li}^{[ij]}
(\varGamma_{ii}^{[jl]}-\varGamma_{il}^{[ji]})^{2}}{(\varGamma_{ii}^{[jl]}\varGamma_{il}^{[ji]})^{2}}, \\
& g_{2}=\sum_{i=1}^{3}\Theta_i+\delta\sum_{1\leq i<l}^3\sum_{j=1}^3\Theta_i\Theta_l\Theta^*_{j,\sigma}
(\Lambda_{ij}^{[lj]}/\varGamma_{ij}^{[lj]})^2+\dfrac{\Theta_1\Theta_2\Theta_{3}}
{\prod^{3}_{i,j=1}\varGamma_{ij}^{[ij]}}
\sum_{1\leq i<j,l\neq i,j}^{3}\Theta^*_{i,\sigma}\Theta^*_{j,\sigma} \nn \\
& \qquad
\cdot\big[a_{ll}\varGamma_{il}^{[jl]}\big(\varGamma_{ii}^{[jj]}(\varGamma_{ji}^{[lj]}
+\varGamma_{ij}^{[li]}-\varGamma_{li}^{[lj]})+\varGamma_{ij}^{[ji]}(\varGamma_{li}^{[lj]}
-\varGamma_{ii}^{[lj]}-\varGamma_{jj}^{[li]})\big)\big]^2,
\end{align}
\end{subequations}
where
\begin{subequations}
\label{nota}
\begin{align}
& \label{alpha}
\varGamma_{ij}^{[lk]}=a_{ij}a_{lk}, \quad
\Lambda_{ij}^{[lk]}=a_{ij}-a_{lk}, \quad
\Theta_j=c_j\rho_j, \quad \tau_{ij}^{[lk]}=\Theta_l\Theta_k\Theta^*_{i,\sigma}\Theta^*_{j,\sigma}, \\
& \Upsilon=\big(\varGamma_{11}^{[12]}\varGamma_{33}^{[23]}(\varGamma_{22}^{[31]}-\varGamma_{21}^{[32]})+\varGamma_{11}^{[13]}
\varGamma_{22}^{[32]}(\varGamma_{21}^{[33]}-\varGamma_{23}^{[31]})
+\varGamma_{12}^{[13]}\varGamma_{21}^{[31]}(\varGamma_{23}^{[32]}-\varGamma_{22}^{[33]})\big)^{2}.
\end{align}
\end{subequations}

\subsubsection{Jordan-block solutions}

If $\Ga_2$ is a Jordan-block matrix as $\Ga_2=\Ga_{\ty{J}}(k_1)$, then we have
\begin{align}
\label{Krtc-Jor}
\br_2=\bF\cdot\bE_{\ty{J}}, \quad \tc_2=\bE^{\st}_{\ty{J}}\cdot\bH, \quad \cd{\bM}_2=\bF\cdot\bG\cdot\bH^*,
\end{align}
with $\bF=\bF_{\ty{J}}(k_1)$, $\bH=\bH_{\ty{J}}(\{c_j\}^{N}_{1})$ and $\bG=\bG_{\ty{J}}(k_1)$.
In this case, $u$ in \eqref{dcmKdV-solu-u} leads to Jordan-block solutions.
In the case of $N=2$, we know that the simplest Jordan-block solution is $u=h_2/h_1$ with
\begin{subequations}
\label{u-Jor-1}
\begin{align}
& h_{1}=a_{11}^{8}+\delta\rho_1\rho^{*}_{1,\sigma}
\big(a_{11}^{4}(-2\sigma c^*_2(a_{11}b-3c_2)
+a_{11}(a_{11}b-2c_2)b^*)+\delta|c_2|^4\rho_1\rho^{*}_{1,\sigma}\big), \\
& h_{2}=a_{11}^{3}\rho_1\big(a_{11}^{5}b+c_2^2\delta\sigma
\big(4c^*_2-\sigma a_{11}b^*\big)\rho_{1}\rho^{*}_{1,\sigma}\big),
\end{align}
\end{subequations}
where $b=c_{1}+c_{2}\kappa$ and $\kappa=2\big(pn/(k^2_1-p^2)+qm/(k_1^2-q^2)\big)$.

\section{Continuum limits of local and nonlocal discrete cmKdV equation}

In this section, we will consider the limiting equations that we retrieve from the
local and nonlocal discrete cmKdV equations \eqref{lcmKdV-1} and \eqref{lcmKdV-2} by
shrinking the lattice grid to a continuous set of values corresponding to spatial and temporal
coordinates. For this purpose, straight continuum limit and
full continuum limit will be performed firstly on the plane-wave factor
\begin{align}
\label{varrho-def}
\rho=\bigg(\frac{p-k}{p+k}\bigg)^{n}\bigg(\frac{q-k}{q+k}\bigg)^{m}\rho^0,
\end{align}
where $\rho^0$ is a constant. Then the above limits will be manipulated on the
equations \eqref{lcmKdV-1} and \eqref{lcmKdV-2} to yield local and nonlocal semi-discrete and continuous cmKdV equations.
The fundamental formula that we will use is the following well-known limit
\begin{align}
\label{dc-re}
\lim\limits_{n\rightarrow\infty}(1+k/n)^n=e^k.
\end{align}
In addition, exact solutions to the resulting local and nonlocal semi-discrete cmKdV equation, as well as
local and nonlocal continuous cmKdV equation, will be discussed. For convenience, we just consider the 1-soliton solution.

\subsection{Straight continuum limit}

The so-called straight continuum limit is to rewrite the first factor in \eqref{varrho-def} in the form \eqref{dc-re}.
We take limit for discrete variable $n$ and lattice parameter $p$ as
\begin{align}
\label{np-1imit}
n\rightarrow \infty,\quad p\rightarrow \infty,\quad \xi=n/p\thicksim O(1).
\end{align}
The plane-wave function \eqref{varrho-def} becomes
\begin{equation}
\label{sd-pwf-scl}
\rho(n,m)\rightarrow \varrho:=\varrho(\xi,m)=\bigg(\frac{q-k}{q+k}\bigg)^{m}\exp(-2k\xi)\varrho_0,
\end{equation}
and $u(n,m)\rightarrow\omega(\xi,m)$. Then from \eqref{lcmKdV-1} and \eqref{lcmKdV-2} we know that
$\mathbb{\omega}(\xi,m)$ satisfies the local and nonlocal semi-discrete cmKdV equation
\begin{align}
\label{cnsd-cmKdV}
\partial_\xi(\omega+\wh{\omega})=2(q^2-\delta(\omega+\wh{\omega})(\omega^*_{\sigma}+\wh{\omega}^*_{\sigma}))^{\frac{1}{2}}(\wh{\omega}-\omega),
\end{align}
which is preserved under transformations $\omega\rightarrow -\omega$ and $\omega\rightarrow \pm i\omega$.
When $\sigma=1$, $\eqref{cnsd-cmKdV}$ is the local semi-discrete cmKdV equation, while when $\sigma=-1$
$\eqref{cnsd-cmKdV}$ is the reverse-$(\xi,m)$ nonlocal semi-discrete cmKdV equation.

We observe that for solutions of the semi-discrete nonlinear equations, they share the same form of those
solutions to the discrete nonlinear counterparts, up to a replacement of plane-wave
factors (cf. \cite{XFZ-TMPH,XZS-SAPM}). Thus, for solutions to the local and nonlocal semi-discrete cmKdV equation
\eqref{cnsd-cmKdV}, we have the following result.
\begin{Thm}
\label{cnsd-cmKdV-Thm}
The function
\begin{align}
\label{cnsd-cmKdV-w}
\omega=\tc_{2}(\bI+\delta\sigma\cd{\bM}_2\cd{\bM}^*_{2,\sigma})^{-1}\br_2,
\end{align}
solves the local and nonlocal semi-discrete cmKdV equation \eqref{cnsd-cmKdV},
provided that the entities satisfy
\begin{subequations}
\label{DES-cnsd}
\begin{align}
& \label{KM-rtc-sd}
\Ga_2 \cd{\bM}_2+\sigma\cd{\bM}_2\Ga^*_2=\varepsilon\br_2\, \tc^{*}_{2}, \\
& \label{r2-solu}
\br_2=(q\bI-\Ga_2)^{m}(q\bI+\Ga_2)^{-m}\mbox{exp}(-2\Ga_{2}\xi)\bC,
\end{align}
\end{subequations}
where $\Ga_2$ is a $N\times N$ constant matrix, $\bC$ is a $N$-th constant column vector and
$\varepsilon^2=\varepsilon^{*^2}=\delta$.
\end{Thm}

The 1-soliton solution to equation \eqref{cnsd-cmKdV} reads
\begin{align}
\label{omega-ss-1}
\omega=\frac{c_1 a_{11}^2\varrho_1}{a_{11}^2+\delta |c_1|^2\varrho_1\varrho^*_{1,\sigma}},
\end{align}
where $\varrho_1=\varrho|_{k\rightarrow k_1}$. The carrier wave of \eqref{omega-ss-1} in the local case is
\begin{align}
\label{omega2-ss-1}
|\omega|^2=\Biggl\{
\begin{array}{ll}
\lambda^2\sech^2(\frac{m}{2}\ln B_1-2 \lambda \xi+\ln C_1),& \text{with} \quad \delta=1, \\
\lambda^2\csch^2(\frac{m}{2}\ln B_1-2 \lambda \xi+\ln C_1),& \text{with} \quad \delta=-1,\\
\end{array}
\end{align}
where $B_1$ and $C_1$ are given in \eqref{ABC1}.

When $\delta=1$, the solution \eqref{omega2-ss-1} is nonsingular and the wave
propagates with initial phase $\ln C_1$, amplitude $\lambda^2$,
top trajectory
\begin{align}
\label{tt}
\xi=(m\ln B_1+2\ln C_1)/(4\lambda),
\end{align}
and velocity $\xi'(m)=\ln B_1/(4\lambda)$. When $\delta=-1$, the solution \eqref{omega2-ss-1} has singularities along with point trace \eqref{tt}.
We illustrate the wave \eqref{omega2-ss-1} with nonsingular in Figure 3.

\vskip30pt
\begin{center}
	\begin{picture}(120,50)
		\put(-180,-23){\resizebox{!}{3.3cm}{\includegraphics{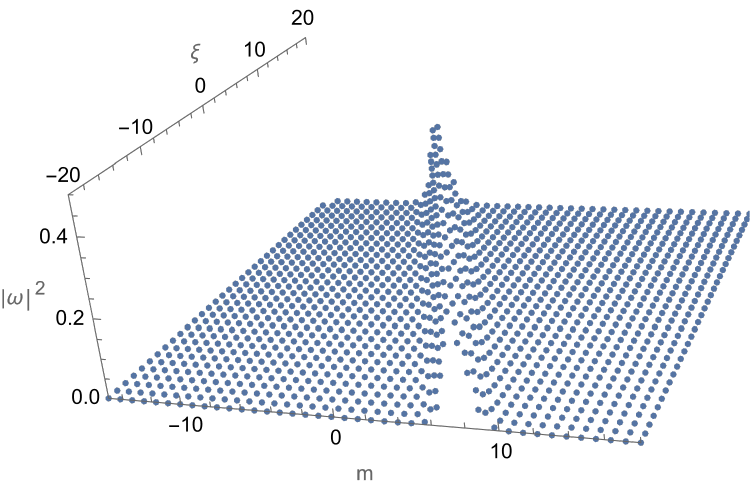}}}
		\put(-10,-23){\resizebox{!}{3.3cm}{\includegraphics{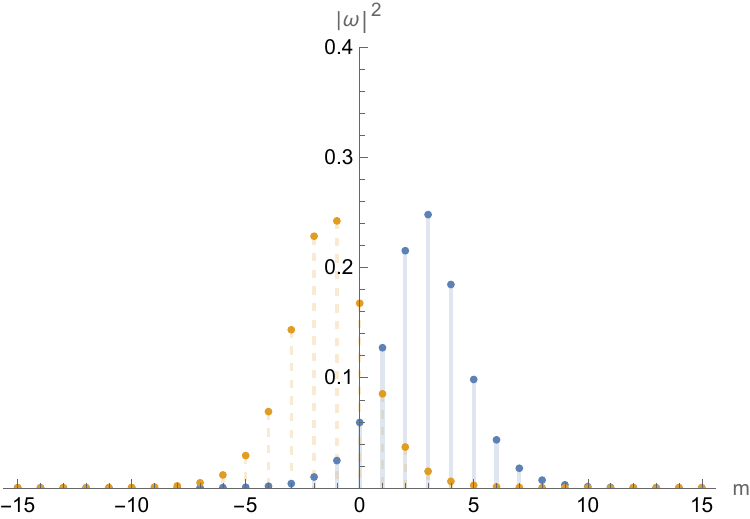}}}
		\put(150,-23){\resizebox{!}{3.3cm}{\includegraphics{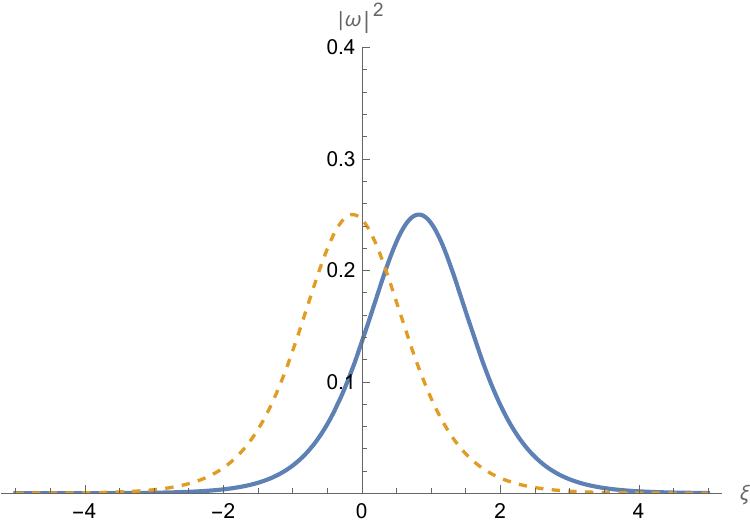}}}
	\end{picture}
\end{center}
\vskip 20pt
\begin{center}
	\begin{minipage}{15cm}{\footnotesize
			\qquad\qquad\quad(a)\qquad\qquad\qquad\qquad\qquad\qquad\qquad\qquad(b) \qquad\qquad\qquad\qquad\qquad\qquad\qquad\qquad (c)
			\\
{\bf Fig. 3} Shape and motion of $|\omega|^2$ given by \eqref{omega2-ss-1} for $\delta=1$, $k_1=(1+i)/2$,
$~\rho_1^0=1$ and $c_1=1+i$. (a) 3D-plot.
			(b) waves in blue and yellow stand for plot (a) at $\xi=-1$ and $\xi=1$, respectively.
			(c) waves from plot (a) at $m=-1$ and $m=1$ respectively shown with solid and dashed hatch.
		}
	\end{minipage}
\end{center}

The wave package in the nonlocal case  has the form
\begin{align}
\label{U-d-ca-II}
|\omega|^2=\dfrac{4\chi^{2}B_1^m\exp(-4\lambda\xi)}{C_2+C_2^{-1}-2\delta\cos(2m\arctan\theta_2-4\chi\xi)},
\end{align}
where $C_2$ and $\theta_2$ are defined by \eqref{theta12-def}.
If $C_2=1$, solution \eqref{U-d-ca-II} has singularities along the point trace
\begin{align}
\label{U-m-singu}
\xi=\Biggl\{
\begin{array}{ll}
\frac{1}{2\chi}(\iota\pi+m\arctan \theta_2),  \quad \delta=1, \\
\frac{1}{4\chi}((2\iota+1)\pi+2m\arctan \theta_2),  \quad \delta=-1, \\
\end{array}
\end{align}
where $\iota\in \mathbb{Z}$. If $C_2 \neq 1$, solution \eqref{U-d-ca-II} is nonsingular and reaches its extrema along the point trace
\begin{align}
\xi=\frac{1}{4\chi}\left(2m\arctan\theta_2+2\iota\pi+\arcsin\frac{(C_2+C_2^{-1})\lambda}{2\delta\sqrt{\lambda^2+\chi^2}}\right),
\end{align}
and the velocity is $\xi'(m)=\arctan \theta_2/(2\chi)$.
We illustrate the wave \eqref{U-d-ca-II} with nonsingular in Figure 4.
\vskip30pt
\begin{center}
	\begin{picture}(120,50)
		\put(-180,-23){\resizebox{!}{3.3cm}{\includegraphics{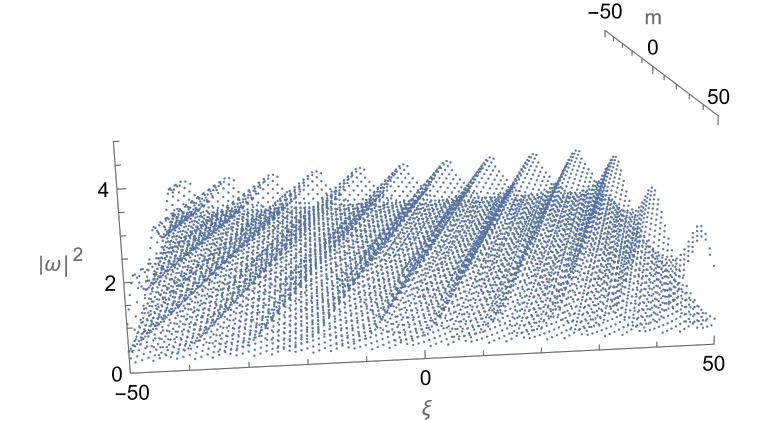}}}
		\put(10,-23){\resizebox{!}{3.3cm}{\includegraphics{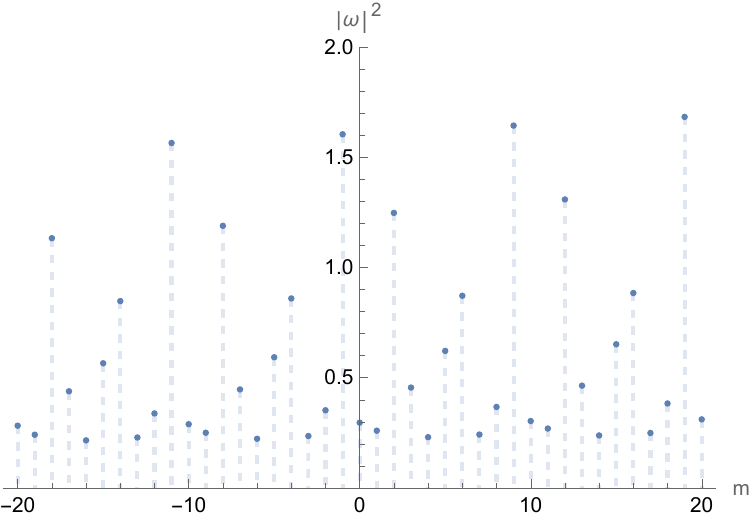}}}
		\put(160,-23){\resizebox{!}{3.2cm}{\includegraphics{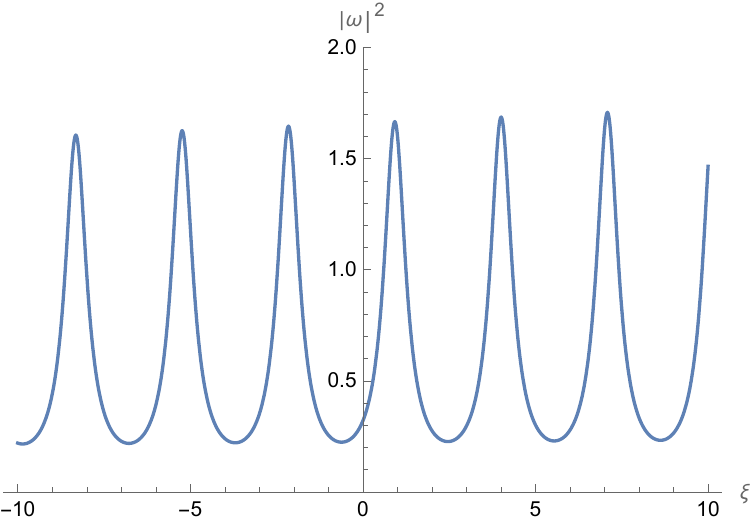}}}
	\end{picture}
\end{center}
\vskip 20pt
\begin{center}
	\begin{minipage}{15cm}{\footnotesize
			\qquad\qquad\qquad\quad(a)\qquad\qquad\qquad\qquad\qquad\qquad\qquad\quad  (b) \qquad\qquad\qquad\qquad\qquad\qquad\qquad\quad (c)
			\\
{\bf Fig. 4} Shape and motion of $|\omega|^2$ given by \eqref{U-d-ca-II} for $\delta=1$, $k_1=-0.001+0.51i$, $~\rho_1^0=1$ and $c_1=1.5$.
			(a) 3D-plot.
			(b) wave from plot (a) at $\xi=1$.
			(c) wave from plot (a) at $m=-1$.
		}
	\end{minipage}
\end{center}

\subsection{Full continuum limit}

This limit can be obtained by investigating the limiting behaviour of the
plane-wave function. From \eqref{sd-pwf-scl}, we have
\begin{align}
\label{sd-pwf-2}
& \varrho_{n}(\xi)=\bigg(\frac{q-k}{q+k}\bigg)^{m}\exp(-2k\xi)\varrho_0 \nn \\
& \qquad\rightarrow \exp\bigg(-2k\xi+\tau q \ln\left(1+\frac{-2k}{q+k}\right)\bigg)\varrho_0 \nn \\
& \qquad\rightarrow \exp\bigg(-2k(\xi+\tau)-\frac{2}{3}\frac{k^3\tau}{q^2}+ O(q^{-4})\bigg)\varrho_0 \nn \\
& \qquad\rightarrow \exp(-2kx-8k^3t+ O(q^{-2}))\varrho_0,
\end{align}
where $\tau=m/q$ and
\begin{equation}
\label{New-var}
x=\xi+\tau,\quad t=\tau/(12q^2),
\end{equation}
which indicates the derivative relations
\begin{align}
\label{der-rel-xt}
\partial_{\tau}=\partial_x+\frac{1}{12q^2}\partial_t,\quad \partial_{\xi}=\partial_x.
\end{align}
After this change of variables and reinterpreting the variable $\omega(\xi,m)$ as
$U(x,t)$, we get at $O(q^{-2})$ the usual local and nonlocal continuous cmKdV equation \eqref{nl-cmKdV}, i.e.,
\begin{align}
\label{cn-cmKdV}
U_t=U_{xxx}+24\delta UU^{*}_{\sigma}U_{x}.
\end{align}

When $\sigma=1$, \eqref{cn-cmKdV} is the local continuous cmKdV equation, while when $\sigma=-1$
\eqref{cn-cmKdV} is the reverse-$(x,t)$ nonlocal continuous cmKdV equation.
This equation is preserved under transformations $U\rightarrow -U$ and $U\rightarrow \pm iU$.
For solutions to the local and nonlocal continuous cmKdV equation \eqref{cn-cmKdV}, we have the following result (cf. \cite{XZ-TMPH}).
\begin{Thm}
\label{so-cnn-cmKdV-Thm}
The function
\begin{align}
\label{cnn-cmKdV-so}
U=\tc_{2}(\bI+\delta\sigma\cd{\bM}_2 \cd{\bM}^*_{2,\sigma})^{-1}\br_2
\end{align}
solves the local and nonlocal continuous cmKdV \eqref{cn-cmKdV}, provided that the entities satisfy
\begin{subequations}
\label{DES-cnn}
\begin{align}
\label{KM-rtcnn-1}
& \Ga_2 \cd{\bM}_2+\sigma\cd{\bM}_2\Ga^*_2=\varepsilon\br_2\, \tc^{*}_{2}, \\
\label{cnn-r2-solu}
& \br_2=\exp(-2\Ga_{2}x-8\Ga_{2}^{3}t)\bC,
\end{align}
\end{subequations}
where $\Ga_2$ is a $N\times N$ constant matrix, $\bC$ is a $N$-th constant column vector and
$\varepsilon^2=\varepsilon^{*^2}=\delta$.
\end{Thm}

The 1-soliton solution to equation \eqref{cn-cmKdV} is expressed as
\begin{align}
\label{U-soli-c-cnn}
U=\frac{c_1a_{11}^2\varsigma_1}{a_{11}^2+\delta |c_1|^2\varsigma_1\varsigma^*_{1,\sigma}},
\end{align}
where $\varsigma_1=\exp(-2k_{1}x-8k_{1}^{3}t)\varsigma_1^0$. The dynamical behaviours of solution \eqref{U-soli-c-cnn}
have been given explicitly in \cite{XZ-TMPH}, which is omitted here.

\section{Conclusion}

In this paper, Cauchy matrix scheme for the discrete AKNS equation is reconsidered.
Starting from the DES \eqref{DES}, we define master functions $\bS^{(i,j)}$ \eqref{Sij-def}, which
possess several properties, including recurrence relations, shift relations, similarity invariance and symmetric property.
Based on these properties, we construct four closed-form equations, i.e., systems (\eqref{u2-albe-eps-re-a}, \eqref{u3-albe-eps-re-a}),
(\eqref{u2-albe-eps-re-a}, \eqref{u3-albe-eps-re-b}),
(\eqref{u2-albe-eps-re-b}, \eqref{u3-albe-eps-re-a}) and (\eqref{u2-albe-eps-re-b}, \eqref{u3-albe-eps-re-b}).
Among these systems, (\eqref{u2-albe-eps-re-a}, \eqref{u3-albe-eps-re-a})
and (\eqref{u2-albe-eps-re-b}, \eqref{u3-albe-eps-re-b}) are `proper' discretization of the AKNS equation since they
can reduce to the local discrete mKdV equation. Although the remain two systems
can not reduce to the local discrete mKdV equation, they admit nonlocal reduction. Since one
cannot obtain Cauchy matrix solutions of the nonlocal real equations in the Cauchy matrix reduction framework \cite{XFZ-TMPH,XZS-SAPM},
we in this paper just focus on the local and nonlocal complex reductions of the resulting discrete AKNS equations. As a consequence,
we derive two local and nonlocal cmKdV equations \eqref{lcmKdV-1} and \eqref{lcmKdV-2}. For these two equations, by solving the
Jordan canonical DES \eqref{dcmKdV-so-nT}, soliton solutions and Jordan-block solutions are constructed. Dynamical behaviours of the 1-soliton solution
are analyzed and illustrated, where the 1-soliton solution of the local discrete cmKdV
equation exhibits the usual bell-type structure, while the 1-soliton solution of the nonlocal discrete cmKdV
equation behaves quasiperiodically. Furthermore, continuum limits, especially straight continuum limit and full continuum limit,
are taken to recovery of the local and nonlocal semi-discrete and continuous cmKdV equations. Consequently, 1-soliton solution to the
local and nonlocal semi-discrete cmKdV equation and its dynamical behaviours are presented and analyzed, respectively.

\vskip 20pt
\section*{Acknowledgments}
This project is supported by Zhejiang Provincial Natural Science Foundation (No. LZ24A010007) and
the National Natural Science Foundation of China (No. 12071432).

\vskip 20pt
\section*{Data Availibility Statement}
Data sharing not applicable to this article as no datasets are generated or
analyzed during the current study.

\vskip 20pt
\section*{Conflict of interest}
There are no conflicts of interest to declare.


\begin{thebibliography}{99}

\bibitem{MPW} H. Markum, R. Pullirsch, T. Wettig,
        Non-hermitian random matrix theory and lattice QCD with chemical potential,
        Phys. Rev. Lett. 83 (1999) 484--487.
\bibitem{LSEK} Z. Lin, J. Schindler, F.M. Ellis, T. Kottos,
        Experimental observation of the dual behavior of PT-symmetric scattering,
        Phys. Rev. A. 85 (2012) 050101.
\bibitem{RMGCSK} C.E. Ruter, K.G. Makris, R. EI-Ganainy, D.N. Christodoulides, M. Segev, D. Kip,
        Observation of parity-time symmetry in optics,
        Nat. Phys. 6 (2010) 192--195.
\bibitem{MMGC} Z.H. Musslimani, K.G. Makris, R. El-Ganainy, D.N. Christodoulides,
        Optical Solitons in PT Periodic Potentials,
        Phys. Rev. Lett. 100 (2008) 030402.
\bibitem{DGPS} F. Dalfovo, S. Giorgini, L.P. Pitaevskii, S. Stringari,
        Theory of Bose-Einstein condensation in trapped gases,
        Rev. Mod. Phys. 71 (1999) 463--512.
\bibitem{Lou-JMP} S.Y. Lou,
        Alice-Bob systems, $\hat{P}-\hat{T}-\hat{C}$ symmetry invariant and symmetry breaking soliton solutions,
        J Math. Phys. 59 (2018) 083507.
\bibitem{Lou-CTP} S.Y. Lou,
        Multi-place physics and multi-place nonlocal systems,
        Commun. Theor. Phys. 72 (2020) 057001.
\bibitem{AbMu-2013} M.J. Ablowitz, Z.H. Musslimani,
        Integrable nonlocal nonlinear Schr\"{o}dinger equation,
        Phys. Rev. Lett. 110 (2013) 064105.
\bibitem{AKNS-1974} M.J. Ablowitz, D.J. Kaup, A.C. Newell, H. Segur,
        The inverse scattering transform-Fourier analysis for nonlinear problems,
        Stud. Appl. Math. 53 (1974) 249--315.
\bibitem{Yang-RW-NLS} B. Yang, J.K. Yang,
        Rogue waves in the nonlocal \({\mathcal {PT}}\)-symmetric nonlinear Schr\"{o}dinger equation,
        Lett. Math. Phys. 109 (2019) 945--973.
\bibitem{JZ-JMAA} J.L. Ji, Z.N. Zhu,
        Soliton solutions of an integrable nonlocal modified Korteweg-de Vries equation through inverse scattering transform,
        J. Math. Anal. Appl. 453 (2017) 973--984.
\bibitem{ALM-JMP} M.J. Ablowitz, X.D. Luo, Z.H. Musslimani,
        Inverse scattering transform for the nonlocal nonlinear Schr\"{o}dinger equation with nonzero boundary conditions,
        J. Math. Phys. 59 (2018) 011501.
\bibitem{Yang-PLA} J.K. Yang,
        General $N$-solitons and their dynamics in several nonlocal nonlinear Schr\"{o}dinger equations,
        Phys. Lett. A 383(4) (2019) 328--337.
\bibitem{AM-mKdV} M.J. Ablowitz, Z.H. Musslimani,
        Integrable nonlocal nonlinear equations,
        Stud. Appl. Math. 139 (2016) 7--59.
\bibitem{RREFM} R.F. Rodriguez, J.A. Reyes, A. Espinosa-Ceron, J. Fujioka, B.A. Malomed,
        Standard and embedded solitons in nematic optical fibers,
        Phys. Rev. E 68 (2003) 036606.
\bibitem{HWLPE} J.S. He, L.H. Wang, L.J. Li, K. Porsezian, R. Erd\'{e}lyi,
        Few-cycle optical rogue waves: complex modified Korteweg-de Vries equation,
        Phys. Rev. E 89 (2014) 062917.
\bibitem{YY-SAPM} B. Yang, J. Yang,
        Transformations between nonlocal and local integrable equations,
        Stud. Appl. Math. 40 (2017) 178--201.
\bibitem{MSZ} L.Y. Ma, S.F. Shen, Z.N. Zhu,
        Soliton solution and gauge equivalence for an integrable nonlocal complex modified Korteweg-de Vries equation,
        J. Math. Phys. 58 (2017) 103501.
\bibitem{CDLZ} K. Chen, X. Deng, S. Lou, D.J. Zhang,
        Solutions of nonlocal equations reduced from the AKNS hierarchy,
        Stud. Appl. Math. 141 (2018) 113--141.
\bibitem{Luo-IST} X.D. Luo,
        Inverse scattering transform for the complex reverse space-time nonlocal modified Korteweg-de Vries
        equation with nonzero boundary conditions and constant phase shift,
        Chaos 29 (2019) 073118.
\bibitem{Ma-JGP} W.X. Ma,
        Inverse scattering and soliton solutions of nonlocal complex reverse-spacetime mKdV equations,
        J. Geom. Phys. 157 (2020) 103845.
\bibitem{GP} M. G\"{u}rses, A. Pekcan,
        Nonlocal modified KdV equations and their soliton solutions by Hirota method,
        Commun. Nonlinear Sci. Numer. Simul. 67 (2019) 427--448.
\bibitem{YDL} L. Li, C.N. Duan, F.J. Yu,
        An improved Hirota bilinear method and new application for a nonlocal integrable complex modified Korteweg-de Vries (MKdV) equation,
        Phys. Lett. A 383 (2019) 1578--1582.
\bibitem{AM-2014} M.J. Ablowitz, Z.H. Musslimani,
		Integrable discrete PT symmetric model,
		Phys. Rev. E 90(3-B) (2014) 1--5.
\bibitem{SMMC} A.K. Sarma, M.A. Miri, Z.H. Musslimani, D.N. Christodoulides,
		Continuous and discrete Schr\"{o}dinger systems with parity-time-symmetric nonlinearities,
		Phys. Rev. E 89 (2014) 052918.
\bibitem{DLZ-AMC} X. Deng, S.Y. Lou, D.J. Zhang,
		Bilinearisation-reduction approach to the nonlocal discrete nonlinear Schr\"{o}dinger equations,
		Appl. Math. Comput. 332 (2018) 477--483.
\bibitem{FZS-IJMPB} W. Feng, S.L. Zhao, Y.Y. Sun,
        Double Casoratian solutions to the nonlocal semi-discrete modified Korteweg-de Vries equation,
        Int. J. Mod. Phys. B 34(5) (2020) 2050021.
\bibitem{CNY} K. Chen, C.N. Na, J.X. Yang,
        Canonical solution and singularity propagations of the nonlocal semi-discrete Schr\"{o}dinger equation,
        Nonlinear Dyn. 111 (2023) 1685--1700.
\bibitem{MZ-JMP} L.Y. Ma, Z.N. Zhu,
        Nonlocal nonlinear Schr\"{o}dinger equation and its discrete version: Soliton solutions and gauge equivalence,
        J. Math. Phys. 57(8) (2016) 83507.
\bibitem{MSZ-AML} L.Y. Ma, S.F. Shen, Z.N. Zhu,
        From discrete nonlocal nonlinear Schr\"{o}dinger equation to coupled discrete Heisenberg ferromagnet equation,
        Appl. Math. Lett. 130 (2022) 108002.
\bibitem{Ger} V.S. Gerdjikov,
		On the integrability of Ablowitz-Ladik models with local and nonlocal reductions,
		J. Phys.: Conf. Ser. 1205 (2019) 012015.
\bibitem{ALM-Nonl} M.J. Ablowitz, X.D. Luo, Z.H. Musslimani,
        Discrete nonlocal nonlinear Schr\"{o}dinger systems: Integrability, inverse scattering and solitons,
        Nonlinearity 33(7) (2020) 3653--3707.
\bibitem{HJN-2016} J. Hietarinta, N. Joshi, F.W. Nijhoff,
        {\it Discrete Systems and Integrablity}, Camb. Univ. Press, Cambridge, (2016).
\bibitem{ZKZ-SIGMA} D.D. Zhang, P.H. Van Der Kamp, D.J. Zhang,
        Multi-component extension of CAC systems,
        SIGMA 16 (2020) 060.
\bibitem{BHQK} T. Bridgman, W. Hereman, G.R.W. Quispel, P.H. van der Kamp,
        Symbolic computation of Lax pairs of partial difference equations using consistency around the cube,
        Found. Comput. Math. 13 (2013) 517--544.
\bibitem{XFZ-TMPH} X.B. Xiang, W. Feng, S.L. Zhao,
        Local and nonlocal complex discrete sine-Gordon equation solutions and contonuum limits,
        Theor. Math. Phys. 211 (2022) 758--774.
\bibitem{XZS-SAPM} X.B. Xiang, S.L. Zhao, Y. Shi,
        Solutions and continuum limits to nonlocal discrete sine-Gordon equations: bilinearization reduction method,
        Stud. Appl. Math. 15(4) (2023) 1274--1303.
\bibitem{ZXS-MMAS} S.L. Zhao, X.B. Xiang, S.F. Shen,
        Solutions and continuum limits to nonlocal discrete modified Korteweg-de Vries equations,
        Math. Meth. Appl. Sci. (2024), DOI 10.1002/mma.9895.
\bibitem{XZ-TMPH} H.J. Xu, S.L. Zhao,
        Cauchy matrix solutions to some local and nonlocal complex equations,
        Theor. Math. Phys. 213 (2022) 1513--1542.
\bibitem{FZ-ROMP} W. Feng, S.L. Zhao,
        Cauchy matrix type solutions for the nonlocal nonlinear Schr\"{o}dinger equation,
        Rep. Math. Phys. 84 (2019) 75--83.
\bibitem{NAJ-JPA} F.W. Nijhoff, J. Atkinson, J. Hietarinta,
        Soliton solutions for ABS lattice equations: I. Cauchy matrix approach,
        J. Phys. A: Math. Theor. 42 (2009) 404005.
\bibitem{ZZ-SAPM} D.J. Zhang,  S.L. Zhao,
        Solutions to ABS lattice equations via generalized Cauchy matrix approach,
        Stud. Appl. Math. 131 (2013) 72--103.
\bibitem{FA} A.S. Fokas, M.J. Ablowitz,
        Linearization of the Korteweg-de Vries and Painlev\'{e} II equations,
        Phys. Rev. Lett. 47 (1981) 1096--1100.
\bibitem{NQC} F.W. Nijhoff, G.R.W. Quispel, H.W. Capel,
        Direct linearization of nonlinear difference-difference equations,
        Phys. Lett. 97A (1983)  125--128.
\bibitem{Syl} J. Sylvester,
        Sur l'equation en matrices $px=xq$,
        C. R. Acad. Sci. Paris 99 (1884) 67--71.
\bibitem{ZS-ZNA} S.L. Zhao, Y. Shi,
        Discrete and semi-discrete models for AKNS equation,
        Z. Naturforsch. A 72 (2017) 281--290.
\bibitem{LQYZ-SAPM} S.S. Li, C.Z. Qu, X.X. Yi, D.J. Zhang,
        Cauchy matrix approach to the SU(2) self-dual Yang-Mills equation,
        Stud. Appl. Math. 148 (2022) 1703--1721.
\bibitem{LQZ-PD} S.S. Li, C.Z. Qu, D.J. Zhang,
        Solutions to the SU($N$) self-dual Yang-Mills equation,
        Physica D 453 (2023) 133828.
\bibitem{ZZN-SAPM} D.J. Zhang, S.L. Zhao, F.W. Nijhoff,
        Direct Linearization of extended lattice BSQ systems,
        Stud. Appl. Math. 129 (2012) 220--248.
\bibitem{Zhao-JNMP} S.L. Zhao,
        A discrete negative AKNS equation: generalized Cauchy matrix approach,
        J. Nonlinear Math. Phys. 23(4) (2016) 544--562.
\bibitem{Zhao-ROMP} S.L. Zhao,
        The Sylvester equation and integrable equations: The Ablowitz-Kaup-Newell-Segur system,
        Rep. Math. Phys. 82(2) (2018) 241--263.
\bibitem{Zhao-JDEA} S.L. Zhao,
        Discrete potential Ablowitz-Kaup-Newell-Segur equation,
        J. Differ. Equ. Appl. 25(8) (2019) 1134--1148.
\bibitem{ZZS-TMPH} S. Zhang, S.L. Zhao, Y. Shi,
        Discrete second-order Ablowitz-Kaup-Newell-Segur equation and its modified form,
        Theor. Math. Phys. 210(3) (2022) 304--326.
\bibitem{Z-KdV-2006} D.J. Zhang,
        Notes on solutions in Wronskian form to soliton equations: KdV-type,
        \textit{arXiv:nlin.SI/0603008}, (2006).

\end{thebibliography}
\end{document}